\documentclass[journal,comsoc]{IEEEtranTCOM}
\usepackage{algorithm,amsbsy,amsmath,amssymb,epsfig,bbm,mathrsfs,multirow,amsthm,algpseudocode}
\usepackage{subcaption,float}
\usepackage{graphicx}
\usepackage{hyperref}
\usepackage{bm,url}
\graphicspath{ {./Figures/System_model}{./Figures/Simulation_Results/}}
\DeclareMathOperator*{\argmax}{\arg\!\max}
\newtheorem{theorem}{Theorem}
\newtheorem{definition}{Definition}

\newtheorem{lemma}{Lemma}

\normalsize
\begin{document}

\title{Time Scheduling and Energy Trading for Heterogeneous Wireless-Powered and Backscattering-based IoT Networks}
\author{Ngoc-Tan~Nguyen,~\IEEEmembership{Student Member,~IEEE,}
	Diep~N.~Nguyen,~\IEEEmembership{Senior Member,~IEEE,}
	Dinh~Thai~Hoang,~\IEEEmembership{Member,~IEEE,}
	Nguyen Van Huynh,~\IEEEmembership{Student Member,~IEEE,}
	Eryk~Dutkiewicz,~\IEEEmembership{Senior Member,~IEEE,}
	Nam-Hoang~Nguyen, and~Quoc-Tuan~Nguyen
	\thanks{N.~T.~Nguyen is with the School of Electrical and Data Engineering, University of Technology Sydney, Sydney, NSW 2007, Australia and JTIRC, VNU University of Engineering and Technology, Vietnam National University, Hanoi, Vietnam (e-mail: tan.nguyen@student.uts.edu.au).}
	\thanks{D.~T.~Hoang, D.~N.~Nguyen, and E.~Dutkiewicz are with the School of Electrical and Data Engineering, University of Technology Sydney, Sydney, NSW 2007, Australia (email: hoang.dinh@uts.edu.au, diep.nguyen@uts.edu.au, eryk.dutkiewicz@uts.edu.au).}%
	\thanks{N.~H.~Nguyen and Q.~T.~Nguyen are with the JTIRC, VNU University of Engineering and Technology, Vietnam National University, Hanoi, Vietnam (email: hoangnn@vnu.edu.vn, tuannq@vnu.edu.vn).}
	\thanks{An abridged version of this paper will be presented at the IEEE Globecom Conference, Dec, 2020~\cite{Tan2020}.}
}%

\maketitle

\begin{abstract}
Future IoT networks consist of heterogeneous types of IoT devices (with various communication types and energy constraints) which are assumed to belong to an IoT service provider (ISP). 
To power backscattering-based and wireless-powered devices, the ISP has to contract with an energy service provider (ESP) via a dedicated power beacon. This article studies the strategic interactions between the ISP and its ESP and their implications on the joint optimal time scheduling and energy trading for heterogeneous backscattering-based and wireless-powered devices. To that end, we propose an economic framework using the Stackelberg game to maximize the network throughput and energy efficiency of both the ISP and the ESP. Specifically, the ISP leads the game by sending its optimal service time  and energy price request (that maximizes its profit) to the ESP. The ESP then optimizes and supplies the transmission power which satisfies the ISP's request (while maximizing ESP's utility). To obtain the \textit{Stackelberg equilibrium} (SE) for the proposed Stackelberg game, we apply a backward induction technique which first derives a closed-form solution for the ESP. Then, to tackle the non-convex optimization problem for the ISP, we leverage the \textit{block coordinate descent} and \textit{convex-concave procedure} techniques to design two partitioning schemes (i.e., partial adjustment (PA) and joint adjustment (JA)) to find the optimal energy price and service time that constitute local SEs. Numerical results reveal that by jointly optimizing the energy trading and the time allocation for heterogeneous IoT devices, one can achieve significant improvements in terms of the ISP's profit compared with those of conventional transmission methods, e.g., bistatic backscatter and harvest-then-transmit communication methods. Different tradeoff between the ESP's and ISP's profits and complexities of the PA/JA schemes can also be numerically tuned. Simulations also show that the obtained local SEs approach the socially optimal welfare when the ISP's benefit per transmitted bit is higher than a given threshold.


\end{abstract}


\section{Introduction}


Frequent recharging/replacing batteries for a massive number of IoT devices (in billions) can be costly, inconvenient, and impractical in some cases (e.g., biomedical implants)~\cite{Derrick2019}. Recent advances in wireless-powered and backscattering communications can alleviate such dependence on battery or even enable the next generation of IoT devices that are battery-free. For the former, a \textit{Harvest-then-transmit} (HTT)~\cite{Ju2014HTT}-\cite{Salem2016} protocol that consists of two phases, i.e., harvesting energy from surrounding radio frequency (RF) signals and active transmissions, can be employed. 
The latter, backscatter communications technology, allows IoT devices to transmit/backscatter information to their receivers by reflecting (instead of actively transmitting) RF signals. These RF signals can come from either dedicated (e.g., card readers) or ambient (e.g., FM, TV stations) sources. There are three typical types of backscatter communications including monostatic, bistatic, and ambient backscatter communications~\cite{Huynh2018}-\cite{Liu2013}. Future IoT networks may consist of these heterogeneous types of IoT devices (with different energy constraints/requirements) which belong to one or different IoT service providers (ISP). 
To ``power" backscattering-based and wireless-powered devices, an ISP has to contract with an energy service provider (ESP) via a dedicated power beacon (PB).

In a WPBC network, a wireless-powered device (WPD) is designed to perform either backscatter communications (i.e., passive transmissions) or transmissions using its RF circuit (i.e., active transmissions) and the energy harvested from the PB. As such, the performance of the WPBC system not only depends on the scheduling or time allocated for energy harvesting, passive, and active transmission operations of IoT devices~\cite{Wang2019},~\cite{Hoang2017Stackelbergame},~\cite{Chen2019} but also the energy contract with the ESP. Most existing works on the WPBC optimize time allocation for IoT devices' operations under the time-division multiplexing (TDM) framework with the assumption of homogeneous IoT devices~\cite{Wang2019}-\cite{Wang2018Stackelberg}. In practice, however, various types of IoT devices with different hardware capabilities and configurations, e.g., performing backscattering, HTT or both can coexist, which have not been considered in the literature. This article studies the strategic interactions between the ISP and its ESP and their implications on the joint optimal energy trading as well as time scheduling for heterogeneous WPBC (HWPBC) networks.

To that end, we propose an economic framework based on the Stackelberg game to jointly maximize the network throughput and energy efficiency of both the ISP and the ESP. The cost function of the ESP is more open and predictable to the ISP but not the vice versa. This is because each ISP has its own set of heterogeneous IoT devices with different hardware capabilities and constraints, and their operating parameters (e.g., the scheduling time) are yet available but to be optimized. Moreover, in practice, an ISP often has more than one option to select an ESP. For that, it has the advantage to initiate/lead the game. In fact, unlike general non-cooperative games, the utility of the leader under a \emph{Stackelberg equilibrium} (SE) is greater or at least equal to that at any \emph{Nash equilibrium} (NE). As such, the ISP can proactively select an energy service from the ESP (i.e., the transmission power of the PB) then leads the game by sending its optimal service time and energy price request (that maximizes its profit) to the ESP. The ESP, as the follower, then finds the optimal transmission power of the PB (i.e., the energy service) based on the offered energy price and service time from the ISP to maximize its benefits. Specifically, to capture the profit of the ESP, we adopt a practical price model for energy generation cost, called the quadratic model~\cite{Mohsenian2010}-\cite{Li2014Joint}. We then derive a closed-form for the optimal transmission power of the PB (ESP) based on the offered price and requested service time from ISP. The profit function of the ISP is defined as the difference between the profit from providing the data services and the energy cost. However, the profit maximization of the ISP is a non-concave problem with respect to the requested energy price and operation times of the PB and IoT devices. Moreover, these variables are strongly coupled, making the non-concave optimization problem of the ISP more challenging. 

To tackle the profit maximization problem of the ISP, we propose two partitioning schemes, called partial adjustment (PA) and joint adjustment (JA) schemes. Using the \textit{block coordinate descent} (BCD) technique~\cite{Tseng2001}, PA and JA schemes find the offered energy price and service time of the PB in an alternating and simultaneous manner, respectively. These schemes produce various strategies for the ISP with different impacts on the ESP's profit. Specifically, for the PA scheme, the iterative algorithm solves three sub-problems with respect to the requested price, service time of the PB, and scheduling times of the IoT devices at each iteration. Whilst the JA scheme splits the original problem into two sub-problems, in which one jointly optimizes the requested price and service time of the PB, and the other optimally allocates the operation times for IoT devices. Then, we adopt the \textit{convex-concave procedure} (CCCP) technique~\cite{Yuille2001} to address the joint sub-problem in the JA scheme. The proposed schemes guarantee to always achieve a local SE. For performance comparisons, we implement simulations to evaluate the profits of the ISP achieved by the proposed Stackelberg game approach (SGA) for heterogeneous IoT devices and other conventional transmission methods (i.e., bistatic backscatter communication mode (BBCM)~\cite{Hilliard2015} and HTT communication mode (HTTCM)~\cite{Ju2014HTT}). Numerical results show that the proposed SGA can outperform other conventional transmission methods in all simulation settings. 
Furthermore, to evaluate the efficiency of local SE, we use the concept of  \textit{Price of Anarchy} (PoA) ratio~\cite{Roughgarden2015},~\cite{Elias2013Joint}. PoA is the ratio of \emph{social welfare} (defined as the sum of the profits of the ISP and ESP) under the worst local SE to that when both the ISP and the ESP fully cooperate (to maximize the social welfare). Via simulations, we observe that the obtained SEs approach the socially optimal welfare when the ISP's benefit per bit exceeds a given threshold. The major contributions of this paper are summarized as follows:
\begin{itemize}
	
	\item 
	We propose a practical economic framework between the ISP and the ESP and study its implications on the joint optimal time scheduling and energy trading for heterogeneous backscattering-based and wireless-powered devices. 
	
	\item 
	We investigate the optimal strategies of the ESP and the ISP as well as the optimal time scheduling for all devices, captured by a local \textit{Stackelberg equilibrium} (SE) of the proposed game. In particular, two schemes (i.e., the PA and JA schemes) performing iterative algorithms based on the BCD and CCCP techniques are proposed to address the non-concave optimization problem of the ISP. These schemes offer different tradeoffs between their complexities and profits for both the ISP and ESP. Moreover, the iterative algorithms are guaranteed to converge to the locally optimal solutions.
	 
	\item 
	We further study the efficiency of the SE through the concept of price of anarchy (PoA). We observe that when the ISP's benefit per bit exceeds a given threshold, the obtained SE approaches the socially optimal welfare (achieved when the ISP and the ESP cooperatively maximize the \textit{social welfare}).
	
	\item We conduct intensive simulations to numerically study the performance and complexity tradeoff for various practical setups. Simulations show that the proposed framework always outperforms other conventional methods in terms of the ISP's profit. 

\end{itemize}

The rest of the paper is organized as follows. Section~\ref{section:system_model} presents the system model. Section~\ref{section:EnergyTrading} formulates the Stackelberg game for joint energy trading and time scheduling. Two iterative algorithms are then proposed in Section~\ref{section:IterativeAlgorithms} to find the local SE. The efficiency of \textit{Stackelberg game} is next analyzed in Section~\ref{section:Efficiency_SE}. We conduct and discuss simulations in Section~\ref{section:NumericalResults} to validate the theoretical derivations. Finally, Section~\ref{section:Conclusions} concludes the paper.

\section{System Model}
\label{section:system_model}
\subsection{Network Setting}
As illustrated in Fig.~\ref{fig:system_model}, we consider a HWPBC network in which wireless-powered and backscattering devices owned by an ISP are powered via a PB of an ESP. For the ISP, we consider three types of low-cost IoT devices with different hardware configurations that can support two functions, i.e., the BBCM and/or HTTCM. The set of active wireless-powered IoT devices (AWPDs) that are equipped with energy harvesting and wireless transmission circuits is denoted by ${\mathcal{A}\rm{ }} \!\buildrel \Delta \over =\! \{\text{AWPD}_a| \forall a \!=\! \{1, \dots, A\}\!\}$. With this configuration, the AWPDs can operate in the HTTCM only. In addition, we denote ${\mathcal{P}\rm{ }} \!\!\buildrel \Delta \over = \!\!\{\text{PWPD}_p| \forall p \!=\! \{1, \dots, P\}\!\}$ as the set of passive wireless-powered IoT devices (PWPDs) that are designed with a backscattering circuit to perform the BBCM only. Finally, hybrid wireless-powered IoT devices (HWPDs) are equipped with hardware components to support both the HTTCM and BBCM. The set of HWPDs is denoted as ${\mathcal{H}\rm{ }}\!\! \buildrel \Delta \over =\! \{\text{HWPD}_h| \forall h\!=\!\{1, \dots, H\}\!\}$.

\begin{figure}[t]
	\centering
	\includegraphics[scale=1.1]{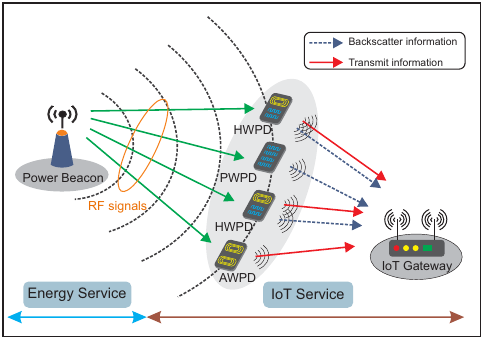}
	\caption{System model of the HWPBC network.}
	\label{fig:system_model}
\end{figure}

The ISP operates over two consecutive working periods of the PB, i.e., \emph{emitting period} $\beta$ and \emph{sleeping period} $(1 \!-\! \beta)$ as illustrated in Fig.~\ref{fig:time_frame}. 
For simplicity and efficiency in time resource allocation for multiple IoT devices, the TDMA mechanism is adopted to avoid collisions among transmissions. We denote $\bm{\theta} \buildrel \Delta \over = \left( \!{{\theta_1}, \ldots, {\theta_p}, \ldots,{\theta_{P}}} \!\right)^\mathsf{T}$ and $\bm{\tau} \buildrel \Delta \over = \left(\! {{\tau_1}, \ldots, {\tau_h}, \ldots,{\tau_{H}}}\! \right)^\mathsf{T}$ as the backscattering time vectors for the PWPDs and HWPDs in the emitting period of the PB, respectively. Similarly, $\bm{\nu}  \buildrel \Delta \over = \left( {{\nu_1}, \ldots ,{\nu_a}, \ldots, {\nu_{A}}} \right)^\mathsf{T}$ and $\bm{\mu} \buildrel \Delta \over = \left( {{\mu_1}, \ldots, {\mu_h}, \ldots,{\mu_{H}}} \right)^\mathsf{T}$ are the transmission time vectors for AWPDs and HWPDs in the idle period of the PB, respectively. When the PB is in the emitting period, it transmits unmodulated RF signals, and thus the IoT devices (i.e., PWPDs and HWPDs) with the backscatter circuits can passively transmit their data by backscattering such signals~\cite{Bletsas2009},~\cite{Kimionis2004}. Meanwhile, the AWPDs and HWPDs equipped with energy harvesting circuits can harvest energy to support their active transmissions in the sleeping period of the PB. Note that, $\text{AWPD}_a$ can harvest energy in the entire emitting period (i.e., $\beta$), while the harvesting time of $\text{HWPD}_h$ is $(\beta \!-\! \tau_h)$ because it must backscatter in time slot $\tau_h$. In the sleeping period of the PB, the AWPDs and HWPDs can perform active transmissions to convey their data to the gateway based on the TDMA protocol.

\subsection{Network Throughput Analysis}
The network throughput (denoted by $R_{sum}$) of communications between the IoT devices and gateway is defined as the total information bits decoded successfully at the gateway over the two periods of the PB.

\begin{figure*}[t]
	\centering
	\includegraphics[scale=0.54]{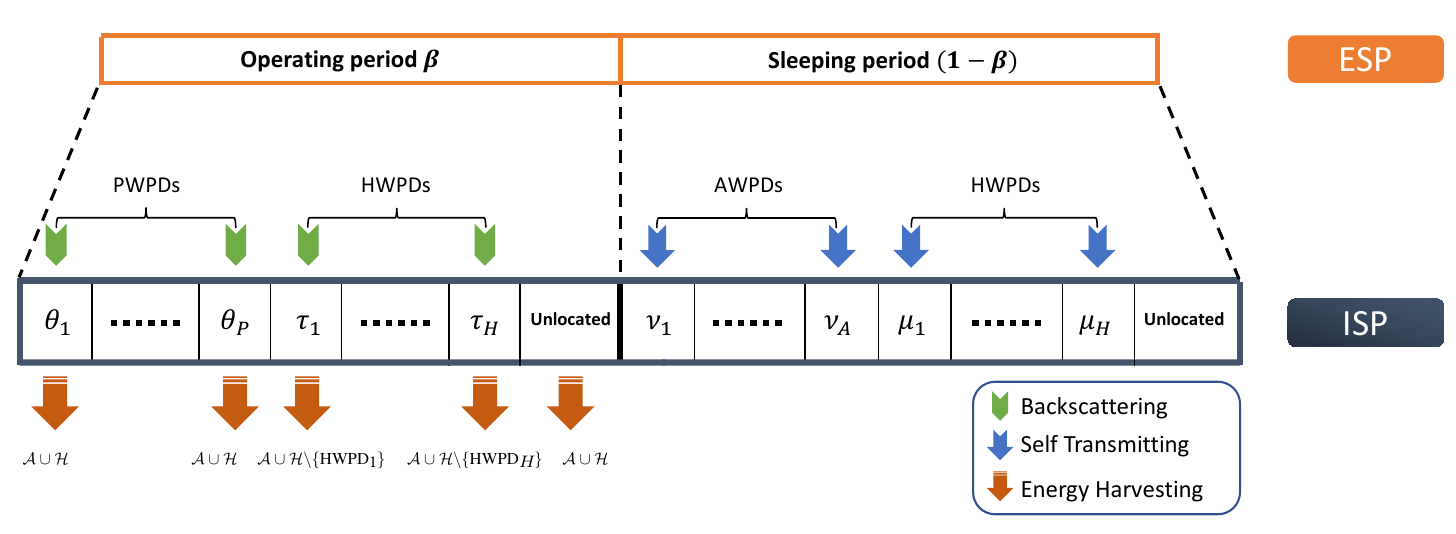}
	\caption{Normalized time frame of the HWPBC network.}
	\label{fig:time_frame}
\end{figure*}

\subsubsection{Emitting period of the PB}
In this period, PWPDs and HWPDs can backscatter the RF signals from the PB to transmit their information. We assume that the PWPDs and HWPDs implement backscatter frequency-shift keying (FSK), or binary FSK to gain more 3 dB in the receiver performance than the classic FSK~\cite{Kimionis2004},~\cite{Hilliard2015}. The power beacon transmits a continuous sinusoid wave of frequency $F_c$ with the complex baseband equivalent as follows:
\begin{equation}
	c\left( t \right) = \sqrt {2{P_S}} {e^{ - j\left( {2\pi \Delta Ft + \Delta \varpi } \right)}},
\end{equation} 
where the $P_S$ is the transmission power of the PB, $\Delta F$ and $\Delta \varpi$ are the frequency and phase offsets, respectively, between the PB and the IoT gateway. 

In the system under consideration, there are three types of communication links: (1) the links from the PB to the IoT devices, (2) the links from the IoT devices to the IoT gateway, and (3) the link from the PB to the IoT gateway which suffer flat fading due to the low bit rate of backscatter communications~\cite{Hilliard2015}. Since the communication ranges of IoT networks are limited, we can consider light-of-sight (LOS) environments in this paper, thus the channel gains for three aforementioned links are given by:
\begin{equation}
{g_{c}} \!=\! \frac{{{G_B}{G_D}{\lambda ^2}}}{{{{\left( {4\pi {d_{BD}}} \right)}^2}}},{g_{d}} \!=\! \frac{{{G_D}{G_G}{\lambda ^2}}}{{{{\left( {4\pi {d_{DG}}} \right)}^2}}},{g_{g}} \!=\! \frac{{{G_B}{G_G}{\lambda ^2}}}{{{{\left( {4\pi {d_{BG}}} \right)}^2}}},
\end{equation}
where $g_c$, $g_d$, and $g_g$ are the channel gains of links (1), (2), and (3), respectively. $G_B$, $G_D$, $G_G$ denote the antenna gains of the PB, IoT devices, and IoT gateway, respectively. $\lambda$ is the wavelength of the RF signals. $d_{BD}$, $d_{DG}$, and $d_{BG}$ are the communication distances of three aforementioned links.  
IoT devices are irradiated by the RF unmodulated signal $c(t)$. Then, the baseband scatter waveform at the IoT devices is written as:
\begin{equation}
	x\left( t \right) = \eta  {u_i}\left( t \right){\sqrt {{g_{c}}}} \;{c\!\left( t \right)}, \quad i \in \left\{ {0,1} \right\},
\end{equation}
where $\eta$ is the attenuation constant of the reflected waveform depending on the backscattering efficiency. For the binary FSK modulation, we consider two distinct load values $\Gamma_i$ with different rates $F_i$ to represent bits $u_i \in \{0, 1\}$, thus the baseband backscatter FSK waveform $u_i(t)$ models the fundamental frequency component of a $50\%$ duty cycle square waveform of frequency $F_i$ and random initial phase ${\Phi _i} \in \left[ {0,2\pi } \right)$:
\begin{equation}
	{u_i}\left( t \right) = {u_0} + \frac{{{\Gamma _0} - {\Gamma _1}}}{2}\frac{4}{\pi }\cos \left( {2\pi {F_i}t + {\Phi _i}} \right),\quad i \in \left\{ {0,1} \right\}
\end{equation}
where $u_0 = \left(A_s - \frac{{{\Gamma _0} + {\Gamma _1}}}{2}\right)$ with $A_s$ is a complex-valued term related to the antenna structural mode~\cite{Bletsas2010}.

The IoT gateway receives both the RF unmodulated signal directly from the PB and the backscattered signals from the IoT devices. Thus, the received baseband signal at the IoT gateway for duration $T$ of a single bit $u_i \in \{0, 1\}$ is given by~\cite{Choi2015Backscatter}:
\begin{equation}
\begin{aligned}
y\left(t\right) &\!=\! \sqrt {{g_{g}}} \;c\!\left( t \right) + \sqrt {{g_{d}}}\; x\!\left( t \right) + n\!\left( t \right) \\
&\!=\! \sqrt {\!2{P_S}} \Bigl\{{ \!\sqrt {\!{g_{g}}}} \!+\! \eta \sqrt {\!{g_{c}}} \sqrt {\!{g_{d}}} {u_0}\\
&{ \quad \!+  \eta \sqrt {\!{g_{c}}}\sqrt {\!{g_{d}}} \frac{2}{\pi }\!\left( {{\Gamma _0} \!-\! {\Gamma _1}}\! \right)\!\cos\! \left( \!{2\pi {F_i}t \!+\! {\Phi _i}}\! \right)} \!\Bigr\} \!\!+\! n(t),
\end{aligned}
\end{equation}
where $n(t)$ is the channel noise.
Before the maximum-likelihood estimation is implemented, carrier frequency offset and removing the \textit{direct current} value from the received signal $y(t)$ are carried out at the IoT gateway~\cite{Kimionis2004},~\cite{Hilliard2015}. The received signal $y(t)$ is then rewritten as follows:
\begin{equation}
	y\!\left( t \right) \!=\! \eta \sqrt {2{P_S}} \sqrt {{g_{c}}} \sqrt {{g_{d}}} \frac{2}{\pi }\!\left( {{\Gamma _0} \!-\! {\Gamma _1}} \right)\!\cos \!\left( {2\pi {F_i}t \!+\! {\Phi _i}} \right).
\end{equation}
Thus, the received power at the IoT gateway is given by:
\begin{equation}
	{P_R^{bb}} = {\eta^2}{g_{c}}{g_{d}}\frac{4}{{{\pi ^2}}}{\left( {{\Gamma _0} - {\Gamma _1}} \right)^2}{P_S}.
\end{equation}
The achievable rate of backscatter communications is given by:
\begin{equation}
	\label{eq: achievable_rate}
	W = \Omega_B {\log _2}\!\left( \!{1 + \frac{{\zeta {P_R^{bb}}}}{{{N_0}}}} \right),
\end{equation}
where $\Omega_B$ is the bandwidth of the unmodulated RF signal, $\zeta$ is the performance gap reflecting real modulation, and $N_0$ is the power spectral density (psd) of the channel noise.
We denote $W_p$ and $W_h$ to be the achievable rates of the $\text{PWPD}_p$ and $\text{HWPD}_h$ calculated in~\eqref{eq: achievable_rate}, respectively. Finally, the total throughput obtained by the AWPDs and HWPDs in the emitting period of the PB is determined as follows:
\begin{equation}
\begin{aligned}
{R^{bb}} &\!=\! \sum\limits_{p = 1}^P {{W_p}} {\theta _p} + \sum\limits_{h = 1}^H {{W_h}} {\tau _h} \\
&\!=\! \sum\limits_{p = 1}^P \!\Omega_B  {\theta _p}{\log _2}\!\!\left(\!{1 \!+\! {\kappa_p}{P_S} }\! \right) \!+\!\! \sum\limits_{h = 1}^H \!\Omega_B  {\tau _h}{\log _2}\!\!\left(\!{1 \!+\! {\kappa_h}{P_S}} \!\right),
 \end{aligned}
\end{equation}
where ${\kappa _p} = {\zeta}{\eta_p^2}{g_{c,p}}{g_{d,p}}{\left( {{\Gamma _0} - {\Gamma _1}} \right)^2}\frac{4}{{\pi ^2}{N_a^0}}$ and ${\kappa _h} = {\zeta}{\eta_h^2}{g_{c,h}}{g_{d,h}}{\left( {{\Gamma _0} - {\Gamma _1}} \right)^2}\frac{4}{{\pi^2}{N_h^0}}$.

\subsubsection{Sleeping period of the PB}
As mentioned in the previous subsection, only AWPDs and HWPDs are able to communicate with the gateway in this period by using their RF transmission circuits. The amount of harvested energy of the $\text{AWPD}_a$ and $\text{HWPD}_h$ from the PB are calculated as follows:
\begin{equation}
\left\{ {\begin{array}{*{20}{l}}
	{{E_a} = \beta P_{R,a}^B},\\
	{{E_h} = \left( {\beta  - {\tau _h}} \right)P_{R,h}^B},
	\end{array}} \right.
\end{equation}
where $P_{R,a}^B = {\varphi_{a}}{g_{{c},a}}{P_S}$ and $P_{R,h}^B = {\varphi_{h}}{g_{{c},h}}{P_S}$ are the received power at the $\text{AWPD}_a$ and $\text{HWPD}_h$ from the PB, respectively~\cite{BalanisAntenna2012}. $\{\varphi_a, \varphi_h\}$ are the harvesting efficiency coefficients of the $\text{AWPD}_a$ and $\text{HWPD}_h$, respectively. For simplicity, we consider the energy consumption by active transmissions of the AWPDs and HWPDs as the dominant energy consumption and ignore the energy consumed by electronic circuits~\cite{Lyu2019Relay}. Hence, the total amount of harvested energy of the AWPDs and HWPDs is utilized to transmit data in the sleeping period of the PB, and the transmission power of the $\text{AWPD}_a$ and $\text{HWPD}_h$ are $P_a^t \!=\! E_a/{\nu_a}$ and $P_h^t \!=\! E_h/{\mu_h}$, respectively. Then the total throughput $R^{st}$ achieved by active transmissions of the AWPDs and HWPDs in the sleeping period of the PB is formulated by:
\begin{equation}
\begin{aligned}
\label{eq: R^tr}
{R^{st}} &\!\!\!=\!\!\! \sum\limits_{a = 1}^A \!\!{{\nu _a}\Omega_D {{\log }_2}}\!\!\left(\!{\!1 \!\!+\!\! \frac{{\zeta}{g_{{d}\!,a}}{P_a^t}}{{N_a^0}}}\! \right) \!\!+\!\!\! \sum\limits_{h = 1}^H \!{{\mu _h}\Omega_D\! {{\log }_2}}\!\!\left(\!{\!1 \!\!+\!\! \frac{{\zeta}{g_{{d}\!,h}}{P_h^t}}{{N_h^0}}} \!\right)\\
&\!=\!\!\!\sum\limits_{a = 1}^A \!{\nu _a}{\Omega_D}{\log _2}\!\!\left(\!{\!1 \!\!+\!\! {\delta_a}\!\frac{{\beta}{P_S}}{{{\nu _a}}}}\! \!\right) \!\!+\!\! \sum\limits_{h = 1}^H \!{\mu _h}{\Omega_D}{\log _2}\!\!\left[\!{\!1 \!\!+\!\! {\delta_h}\!\frac{{(\!\beta \!-\!\tau_h\!)\!P_S}}{{{\mu _h}}}}\!\!\right]\!\!,
\end{aligned}
\end{equation}
where $\delta_{a} \!=\!\frac{{\zeta}{\varphi_a}{g_{{c},a}}{g_{{d},a}}}{N_0^a}$ and $\delta_{h} \!=\!\frac{{\zeta}{\varphi_h}{g_{{c},h}}{g_{{d},h}}}{N_0^h}$. $\Omega_D$ is the bandwidth for the HTT protocol, and $\{N_a^0, N_h^0\}$ are the noise of the communication channels from the $\text{AWPD}_a$ and $\text{HWPD}_h$ to the gateway, respectively.

Finally, the network throughput ($R_{sum}$) of the ISP can be determined as follows:
\begin{equation}
\begin{aligned}
\label{eq: Rsum1}
&{R_{sum}}\!\left( \bm{\theta},\bm{\nu},\bm{\tau},\bm{\mu} \right) = R^{bb} \!+\! R^{st}  \\
\!\! = \!\!&\sum_{p = 1}^P\! {{\Omega_B}{\theta_p}{\log _2}\!\left({1 \!+\! {\kappa_p}{P_S} }\! \right)} \!+\! \sum_{a = 1}^A \!{{\Omega_D}{\nu_a}{\log_2}\!\!\left(\!\!1 \!+\! {\delta_a}\frac{ {\beta}{P_S}}{\nu_a}\!\!\right)} \\
\!\!+\!\! &\sum_{h = 1}^H \!\!{\left\{\!\!{{\Omega_B}{\tau_h}{\log _2}\!\left({\!1 \!\!+\! {\kappa_h}{P_S} }\! \right)} \!+\! {\Omega_D}{\mu_h}{\log_2}\!\!\left[\!1 \!\!+\! {\delta_h}\!\!\frac{( \beta \!-\! {\tau_h}\!)\!{P_S}}{\mu_h}\!\!\right]\!\!\right\}}\!.
\end{aligned}
\end{equation}
 It is modeled as the achieved profit of the communication service to jointly maximize the benefits of both service providers in the HWPBC network.

\section{Joint Energy Trading and Time Allocation based on Stackelberg Game}
\label{section:EnergyTrading}
Based on the system model given in the Section~\ref{section:system_model}, in this section, we first introduce the Stackelberg game to model the strategic interaction between the ISP and ESP. Then, we derive the strategic behaviors of these service providers which maximize their own profits. 

\subsection{Game Formulation}
\begin{itemize}
	\item \textbf{Leader payoff function}: The achievable benefit of the ISP is defined as follows:
	\begin{equation}
	\label{eq:Leader_func}
	\begin{aligned}
	{\bm{U}_L} \!\left( {{p_l}} , {\beta}, \bm{\psi}  \right) = {p_r}{R_{sum}} - {p_l}{ \beta}{P_S},
	\end{aligned}
	\end{equation}
	where $p_r$ is the benefit per bit transmitted by IoT devices, and $p_l$ is the energy price paid by the ISP to the ESP. The leader maximizes its utility function $\bm{U}_L$ w.r.t. the energy price $p_l$, operation time $\beta$, and time scheduling $\bm{\psi} \buildrel \Delta \over = (\bm{\theta}, \bm{\nu}, \bm{\tau}, \bm{\mu})$.

	\item \textbf{Follower utility function}: In this game, the PB is the follower and it optimizes its transmission power based on the requested energy price and operation time from the ISP. The utility function of the follower is formulated based on its profit obtained from the ISP and its cost incurred during the operation time:
	\begin{equation}
	\label{eq: follower_payoff_func}
	{{\bm{U}}_F}\left( {{P_S}} \right) = {\beta}\left[\; {{p_l}{P_S} - F({{P_S}})} \right],
	\end{equation}
	where $F(x) = {a_m}x^2 + b_{m}x$ is a quadratic function which is applied for the actual energy generation cost of the PB~\cite{Mohsenian2010}-\cite{Li2014Joint}.
\end{itemize}

\subsection{Solution to the Stackelberg Game}
The definition of the \emph{Stackelberg equilibrium} (SE) is stated as follows: 
\begin{definition}
	\label{definition: SEpoint}
	$(P_S^{*}, p_l^*, \beta^*, \bm{\psi}^*)$ is a Stackelberg equilibrium of the above game if  the following conditions are satisfied~{\normalfont\cite{Fudenberg1991}}:
	\begin{equation}
	\left\{ \begin{array}{ll}
	{{\bm{U}}_L}\!\left({{P_S^{*}},p_l^*, \beta^*, \bm{{\psi}}^*} \right) \ge {\bm{U}_L}\!\left({{P_S^{*}},{p_l},\beta, \bm{{\psi}}}\right),\\
	{{\bm{U}}_F}\!\left({{P_S^{*}},p_l^*,\beta^*, \bm{{\psi}}^*} \right) \ge {{\bm{U}}_F}\!\left( {{P_S},p_l^*,\beta^*, \bm{{\psi}}^*} \right).
	\end{array} \right.
	\end{equation}
\end{definition}
We adopt the backward induction technique to obtain the Stackelberg game solution. Firstly, given a strategy of the leader (i.e., the ISP), a unique optimal solution of the follower (i.e., the ESP) can be obtained straightforwardly in a closed-form since the follower's utility is a quadratic function:
\begin{equation}
\label{eq:followerSol}
	{P_S^{*}} = \frac{{{p_l} - b_m}}{{2a_m}}.
\end{equation}
Then, given the optimal transmission power $P_S^{*}$ of the follower, the leader payoff function can be rewritten as in~\eqref{eq: rewritten_leader_func}.
\begin{figure*}[!]
	\begin{equation}
	\begin{aligned}
	\label{eq: rewritten_leader_func}
	{{\bm{U}}_L}\! \left( {{p_l}}, \beta, \bm{\psi} \right) &= {p_r}\Biggl\{ {\sum\limits_{p = 1}^P {{\Omega_B}{\theta_p}{\log _2}\!\left(\!\!{1 \!+\! {\kappa _p}\frac{{\left( {{p_l} - {b_m}} \right)}}{{2{a_m}}} } \!\right)} \!+\!\sum\limits_{a = 1}^A\! {\Omega_D}{{\nu _a}} {{\log }_2}\!\left[\! {1 \!+\! {\delta _a}\frac{{\beta\left( {{p_l} - {b_m}} \!\right)}}{{2{\nu _a}{a_m}}}} \!\right]}\\
	& \quad \quad	{+ \sum\limits_{h = 1}^H \! \left[ \!{{{\Omega_B}{\tau_h}{\log _2}\!\left(\!\!{1 \!+\! {\kappa _h}\frac{{\left( {{p_l} \!-\! {b_m}} \right)}}{{2{a_m}}} }\! \!\right)} \!+\! {\Omega_D}{\mu _h}{{\log }_2} \!\left(\!\! {1 \!+\! {\delta _h}\frac{{( \beta  \!-\! {\tau _h})\!\left( {{p_l} \!-\! {b_m}} \right)}}{{2{\mu _h}{a_m}}}} \!\right)} \!\!\right]}\! \!\Biggr\} \!-\! \frac{{{p_l}\beta\!\left( {{p_l} \!-\! {b_m}} \right)}}{{2{a_m}}}. \quad \quad\quad \quad
	\end{aligned}
	\end{equation}
		\hrulefill
\end{figure*}
The profit maximization problem for the leader is expressed as follows:
\begin{subequations}\label{opt1:main}
	\begin{align}
	&\mathop {\max }\limits_{\left( {{p_l}, \beta, \bm{\psi} } \right)}{{\bm{U}}_L} \!\left( {{p_l}, \beta}, \bm{\psi} \right), \tag{\ref{opt1:main}}\\
	\text{s.t.} \; & 0 \le P_S \le P_{S}^{max}, \label{opt1:a} \\
	& P_i^{min} \le P_i^t \le P_i^{max}, i \in \left\{ {a,h} \right\}, \label{opt1:b}\\
	& E_i^{min} \le E_i \le {E_i^{max}}, i \in \left\{ {a,h} \right\},\label{opt1:c}\\
	& \gamma_i^{bb} \ge \gamma_i^{\min }, i \in \left\{ {p,h} \right\}, \label{opt1:d}\\
	& 0 \le \sum\nolimits_{p = 1}^P {\theta _p} \!+\! \sum\nolimits_{h = 1}^H {\tau _h} \le \beta \le 1, \forall {\theta _p}, \forall {\tau _h} \!\ge\! 0, \label{opt1:e}\\
	&  0 \!\le\!\! \sum\nolimits_{a = 1}^A\! {\nu _a}  \!+\! \sum\nolimits_{h = 1}^H \!{\mu _h} \!\le\! 1 \!-\! \beta \!\le\! 1, \forall {\nu _a}, \forall {\mu _h} \!\ge\! 0.  \label{opt1:f}
	\end{align}
\end{subequations}
where the constraint~\eqref{opt1:a} specifies that the transmission power of the PB, i.e., ${P_S} = \frac{{\left( {{p_l} - {b_m}} \right)}}{{2{a_m}}}$ must satisfy the FCC Rules~\cite{FCC_rules} for unlicensed wireless equipment operating in the industrial, scientific, and medical (ISM) bands. For the IoT devices, the constraint~\eqref{opt1:b} ensures that the transmission power of AWPDs and HWPDs, i.e., $P_a^t = \frac{{{\varphi _a}{g_{c,a}}\beta \left( {{p_l} - {b_m}} \right)}}{{2a_m}{\nu_a}}$ and $P_h^t = \frac{{{\varphi _h}{g_{c,h}}\left( {\beta  - {\tau _h}} \right)\left( {{p_l} - {b_m}} \right)}}{{2{a_m}{\mu_h}}}$, respectively, are sufficient for active communications to the IoT gateway as well as under a threshold. The total energy harvested by the AWPDs and HWPDs in the emitting period of the PB, i.e., ${E_a} = \frac{{{\varphi _a}{g_{c,a}}\beta \left( {{p_l} - {b_m}} \right)}}{{2{a_m}}}$ and ${E_h} = \frac{{{\varphi _h}{g_{c,h}}\left( {\beta  - {\tau _h}} \right)\left( {{p_l} - {b_m}} \right)}}{{2{a_m}}}$, respectively, must be sufficient for their operations, as well as not exceed the capacity of their batteries as represented in the constraints~\eqref{opt1:c}. Furthermore, the \emph{signal-to-noise ratio} (SNR) at the gateway received from PWPDs and HWPDs by backscatter communications, i.e., $\gamma _p^{bb} = \frac{{{\kappa _p}\left( {{p_l} - {b_m}} \right)}}{{2{a_m}}} and \gamma _h^{bb} = \frac{{{\kappa _h}\left( {{p_l} - {b_m}} \right)}}{{2{a_m}}}$, respectively, must satisfy the constraints~\eqref{opt1:d} to guarantee \textit{bit-error-rate} lower than or equal to $10^{-2}$~\cite{Kimionis2004}. Finally, the constraints~\eqref{opt1:e}-\eqref{opt1:f} are time constraints to impose IoT devices working on the proper periods. In particular, the PWPDs and HWPDs must backscatter RF signals in the emitting period and the AWPDs and HWPDs must perform active transmissions in the idle period of the PB.

We then state the existence of a SE for the proposed Stackelberg game in the following Theorem.
\begin{theorem}
	\label{theorem: SE_existance}
	There exists at least a SE $\left(P_S^*, \bm{\chi}^*\right)$ for the proposed Stackelberg game satisfying the Definition~\ref{definition: SEpoint} where ${P_S^{*}}$ is obtained in~\eqref{eq:followerSol} and $\bm{\chi}^* \buildrel \Delta \over = ({p_l^*}, \beta^*, \bm{\psi}^*)$ is the globally optimal solution of the maximization problem~\eqref{opt1:main}.
\end{theorem}
\begin{proof}
	See the Appendix~\ref{App:Theorem1}.
\end{proof}

However, the problem~\eqref{opt1:main} is non-concave due to its non-convex feasible set. Specifically, the constraint~\eqref{opt1:c} is non-convex w.r.t. $\left(p_l, \beta\right)$ (due to its below negative Hessian).
\begin{equation}
\label{eq:hessianMatrix}
	\textbf{M} = \left[ {\begin{array}{*{20}{c}}
		{0}&{1}\\
		{1}&{0}
		\end{array}} \right].
\end{equation}

Moreover, variables in the objective function~\eqref{eq: rewritten_leader_func} and the constraint~\eqref{opt1:b} of the non-concave problem \eqref{opt1:main} are strongly coupled. 
To tackle it, in the next section, we introduce low-complexity iterative algorithms using the BCD technique to obtain the locally optimal solution for the profit optimization problem of the ISP. 

Note that the game with non-concave utility-maximization problem is often referred to as non-convex or non-concave game that is challenging. In this case, one tends to relax the equilibrium concept to quasi-equilibrium~\cite{Siyari2017}-\cite{Scutari2013}. In our case, a quasi-SE (QSE) can be defined as a solution of a variational inequality~\cite{Facchinei2007} equivalent-problem obtained under the \textit{Karush–Kuhn–Tucker} (K.K.T.) optimality conditions of the non-cave problem~\eqref{opt1:main}. However, in our work, we adopt the concept of local SE that is defined as follows~\cite{Daskalakis2018}-\cite{Mazumdar2018}:
\begin{definition}
	\label{definition:localSE}
	Let $S_{\bm{\chi}}$ be the feasible region shaped by the constraints~\eqref{opt1:a}-\eqref{opt1:f}. A pair $\left({P_S^*, \bm{\hat \chi}^*}\right)$ is a local SE of the proposed Stackelberg game if there exists a neighborhood $\hat S_{\bm{\chi}}$ around $\bm{\hat \chi}^*$ so that for all $\bm{\chi} \in \hat S_{\bm{\chi}} \subset S_{\bm{\chi}}$, we have:
	\begin{equation}
	{\bm{U}_L}\left( {{\bm{\hat \chi} ^*}}, {P_S^*} \right) \ge {\bm{U}_L}\left( {\bm{\chi} ,P_S^*} \right).	
	\end{equation}
\end{definition}
It is worth noting that the concept of local SE above is stronger than the concept of QSE as a local SE is a QSE but a QSE is not always a local SE.

\section{Iterative Algorithms to Find the local Stackelberg Equilibrium}
\label{section:IterativeAlgorithms}
In this section, to find the local SE, we propose two partitioning schemes, i.e., the PA and JA which employ the BCD and CCCP techniques to address the non-concave optimization problem (\ref{opt1:main}). The idea is to decompose the original problem into sub-problems that are concave and can be effectively solved in each iteration. The JA scheme can outperform the PA scheme in maximizing the profit of the ISP, while the PA scheme requires less computational resources than the JA scheme.

\subsection{PA Scheme}

This scheme performs an iterative algorithm to partition the variable tuple $\bm{\chi}$ into 3 different blocks of variables, i.e., the energy price $ {p_l}$, the emitting time ${\beta}$, and the scheduling times ${\bm{\psi}} \buildrel \Delta \over =  (\bm{\theta}, \bm{\tau}, \bm{\nu}, \bm{\mu})$. At each iteration, we (i) optimize the energy price ${p_l}^{(n)}$ from the last optimal output $\{{p_l^{(n-1)}},{\beta}^{(n-1)}, {\bm{\psi}}^{(n-1)}\}$; (ii) obtain the emitting time of the PB ${\beta}^{(n)}$ by keeping the $\{{p_l}^{(n)}, {\bm{\psi}}^{(n-1)}\}$ fixed; (iii) and find the optimal scheduling times ${\bm{\psi}}^{(n)}$ of the IoT devices with the fixed ${p_l}^{(n)}$ and ${\beta}^{(n)}$. 
These steps are described in detail as follows:
\subsubsection{Optimal Energy Price Offered for the PB}
In the first step of the algorithm loop, we obtain the optimal requested price $p_l$ based on the optimal solution from the previous step $\{{p_l^{(n-1)}}, {\beta^{(n-1)}}, {{\bm{\psi}}^{(n-1)}}\}$. Note that the time constraints in~\eqref{opt1:main} are eliminated because the time variables are constant and set by the previous optimal vector $\bm{\psi}^{(n-1)}$. Then, the original optimization problem in \eqref{opt1:main} can be transformed into:
\vspace{-2pt}
\begin{subequations}\label{subopt1:main}
	\begin{alignat}{3}
	&\mathop {\max }  \limits_{{p_l}} { G} ({{p_l}} ) , \tag{\ref{subopt1:main}}\\
	\text{s.t.} \quad & 0 \le {p_l} - {b_m} \le 2{a_m}{P_S^{\max }} , \label{subopt1:a} \\
	& P_a^{min} \le {{\frac{{\varphi_a}{g_{c,a}}{\beta^{(\!n-1\!)}}}{2{\nu_a^{(\!n-1\!)}}{a_m}}}\left( {{p_l} - {b_m}} \right)} \le P_a^{max}, \label{subopt1:b}\\
	& P_h^{min} \!\le\! {{\frac{{\!{\varphi_h}{g_{c,h}}\!\!\left(\!\beta^{(\!n-1\!)} \!-\! \tau_h^{(\!n-1\!)}\!\!\right)}\!}{2{\mu_h^{(\!n-1\!)}}{a_m}}}\!\left( {{p_l} \!-\! {b_m}} \!\right)} \!\le\! P_h^{max}, \label{subopt1:c}\\
	& E_a^{min} \le {{\frac{{\varphi_a}{g_{c,a}}{\beta^{(\!n-1\!)}}}{2{a_m}}}\left( {{p_l} - {b_m}} \right)} \le E_a^{max},\label{subopt1:d}\\
	& E_h^{min} \!\le\! {{\frac{{\!{\varphi_h}{g_{c,h}}\!\!\left(\!\beta^{(\!n-1\!)} \!-\! \tau_h^{(\!n-1\!)}\!\right)}\!}{2{a_m}}}\!\!\left( {{p_l} \!-\! {b_m}} \!\right)} \!\le\! E_h^{max}, \label{subopt1:e} \\
	& {{\kappa_p}}\!\left( {{p_l} - {b_m}} \right) \ge {2{a_m}}\gamma_p^{\min }, \label{subopt1:f}\\
	& {{\kappa_h}}\!\left( {{p_l} - {b_m}} \right) \ge {2{a_m}}\gamma _h^{\min }, \label{subopt1:g}
	\end{alignat}	
\end{subequations}
where ${G}{\left(p_l\right)}$ is expressed in~\eqref{eq: G_1}
\begin{figure*}
\begin{equation}
\begin{aligned}
\label{eq: G_1}
{ G} \left( {{p_l}}\right) &= {\sum\limits_{p = 1}^P\!  {{{c_1}{\theta_p^{(n-1)}}}{\log _2}\! \left[ {1  \!+\!  {\kappa_p}{\frac{{\left( {{p_l}  \!-\!  {b_m}} \! \right)}}{2{a_m}}}}  \right]} } \!+\! \sum\limits_{a = 1}^A \! {{c_2}{\nu_a^{(\!n-1\!)}}}{{\log }_2} \!\! \left[ {\!1  \!+\!  {\delta_a}\frac{{\beta^{(\!n-1\!)}}\!{\left({{p_l}\!-\!{b_m}} \! \right)}}{2{\nu_a^{(\!n-1\!)}}{a_m}}} \! \right]\\
& \quad \!+\! \sum\limits_{h = 1}^H  \! \!\left\{ \! {{{c_1}{\tau_h^{(\!n-1\!)}}{\log _2}\!\! \left[\!{1 \!+\! {\kappa_h}\frac{{\left( {{p_l} \!- \!{b_m}} \!\right)}}{2{a_m}} } \! \right]}}\! \!+\! {c_2}{\mu_h^{(n-1)}}{{\log }_2} \!\! \left[ \!{1 \!+\! {\delta_h}\frac{\!{\left(\!\beta^{(\!n-1\!)} \!-\! \tau_h^{(\!n-1\!)}\!\right)}\!{\left( {{p_l} \!- \!{b_m}} \!\right)}}{2{\mu_h^{(\!n-1\!)}}{a_m}}} \! \right] \!\! \right\} \!-\!  \frac{{p_l}{\beta^{(\!n-1\!)}}\!{{\left( {{p_l} \!-\! {b_m}} \right)}}}{2{a_m}}, 
\end{aligned}
\end{equation}
\hrulefill
\end{figure*}
and
${c}_{1} \!=\! {p_r}\Omega_B$,
${c_{2}} \!=\! {p_r}{\Omega_D}$.

\begin{lemma}
	\label{lemma: convex_proof_subopt1}
	The objective function $G$ is a concave function w.r.t. $p_l$ satisfying the linear constraints in \eqref{subopt1:a}-\eqref{subopt1:g}, and the optimal solution for the single variable sub-problem \eqref{subopt1:main} can be obtained by line search methods.
\end{lemma}

\begin{proof}
	The function ${ G} {(p_l)}$ is a sum of logarithmic functions of $p_l$ which has the form of $\log_2({a_t}x+b_t)$ and a quadratic function ${ f} {(p_l\!)} = -\frac{{p_l}{\beta^{(\!n-1\!)}}\!{\left(p_l - b_m\!\right)}}{2{a_m}}$. Intuitively, the logarithmic function $\log_2({a_t}x+b_t)$ is a concave function w.r.t. $x$. Furthermore, the quadratic function ${ f} {(p_l)}$ is also a concave function. Thus, the objective function ${ G}$ is a concave function w.r.t. $p_l$. Since the sub-problem~\eqref{subopt1:main} is a single variable optimization problem which can be solved efficiently by using the line search methods such as the golden section or parabolic interpolation methods~\cite{Lee2016Sum}.
\end{proof}

\subsubsection{Optimal Emitting Time of the PB}
The optimal emitting time $\beta$ of the PB in the \textit{n-}th iteration can be obtained in the second step by solving the following sub-problem:
\vspace{-5pt}
\begin{subequations}\label{subopt2:main}
\begin{alignat}{3}
&\mathop {\max }\limits_{{\beta}} {{\hat G}} \left( \beta \right), \tag{\ref{subopt2:main}}\\
\text{s.t.} \quad & 0 \le {\beta} \le 1, \label{subopt2:a} \\
& P_a^{min} \le \frac{{c_3}{\varphi_a}{g_{c,a}}{\beta}}{\nu_a^{(\!n-1\!)}} \le P_a^{max},  \label{subopt2:b} \\
& P_h^{min} \le \frac{{c_3}{\varphi_h}{g_{c,h}}{\left({\beta - \tau_h^{(\!n-1\!)}}\!\right)}}{\mu_h^{(\!n-1\!)}} \le P_h^{max}, \label{subopt2:c}\\
& E_a^{min} \le {{c_3}{\varphi_a}{g_{c,a}}{\beta}} \le E_a^{max},\label{subopt2:d} \\
& E_h^{min} \le {{c_3}{\varphi_h}{g_{c,h}}{\left({\beta - \tau_h^{(\!n-1\!)}}\!\right)}} \le E_h^{max}, \label{subopt2:e}
\end{alignat}	
\end{subequations}
where
\begin{equation}
\begin{aligned}
&{\hat G} \! \left( \beta  \right) \!=\! {\sum\limits_{a = 1}^A \!{c_2}{\nu_a^{(n-1)}}{{\log }_2}\!\left[ \!{1 \!+\! \frac{{c_3}{\delta_a}{\beta}}{{\nu_a^{(n-1)}}} } \right]} \\
&  \!+\!\! {\sum\limits_{h = 1}^H\! {c_2}{\mu_h^{(\!n-1\!)}}{{\log }_2}\!\!\left[ \!{1 \!+\! \frac{{c_3}{\delta_h}\!{\left(\!{\beta\!-\!\tau_h^{(\!n-1\!)}}\!\right)}}{{\mu_h^{(n-1)}}}}\!\!\right]} \!\!-\! {c_3}{\beta}{p_l^{(\!n\!)}} \!+\! { C},
\end{aligned}
\end{equation}
\begin{equation}
\begin{aligned}
{C} \!=\!\!\! {\sum\limits_{p = 1}^P \!{c_1}{\theta_p^{\left(\!{n - 1}\!\right)}\!{{\log }_2}\!\left[{1 \!+\! {c_3}{\kappa _p}} \right]}} \!\!+\!\!\! {\sum\limits_{h = 1}^H \!\!{c_1}{\tau _h^{\left( \!{n - 1} \!\right)}\!{{\log }_2}\!\left[ {1 \!+\! {c_3}{\kappa _h}}\right]}}, 
\end{aligned}
\end{equation}
and ${c}_{3} \!=\! \frac{{\left(\! {p_l^{(n)} - {b_m}} \!\right)}}{{2{a_m}}}.$ 

Similar to the sub-problem~\eqref{subopt1:main}, the transmission power constraint of the PB, the time constraints of all IoT devices, and the SNR constraints of backscatter devices are always satisfied with the fixed $\{p_l^{(\!n\!)}, {\bm{\psi}}^{(\!n-1\!)}\!\}$, and thus they can be omitted.
\begin{lemma}
	\label{lemma: convex_proof_subopt2}
	The objective function $\hat{G}$ is a concave function w.r.t. $\beta$ satisfying the linear constraints in \eqref{subopt2:a}-\eqref{subopt2:e}, and the optimal solution for the single variable sub-problem \eqref{subopt2:main} can be obtained by line search methods.
\end{lemma}
\begin{proof}
Following the proof of the Lemma~\ref{lemma: convex_proof_subopt1}, the function ${\hat G} {(\beta)}$ is contributed by logarithmic functions forming as $\log_2\!{(\!a_t{x} \!+\! b_t\!)}$, a linear function ${\hat f}{(\beta)} \!=\! -{ c_3}{\beta}{p_l^{(\!n\!)}}$, and a constant $C$. The logarithmic function $\log_2\!{(\!a_t{x} \!+\! b_t\!)}$ is also concave w.r.t. $x$. Thus, the objective function ${\hat G}$ is concave w.r.t. $\beta$. Therefore, the optimal solution of the single variable sub-problem \eqref{subopt2:main} can be also found efficiently by line search methods.
\end{proof}


\subsubsection{Optimal Time Resource Allocation}
In the third step, we investigate the time scheduling $\bm{\psi}^{(n)}$ based on the given $\{p_l^{(n)}, \beta^{(n)}\}$. The original optimization problem \eqref{opt1:main} is simplified as:
\begin{subequations}\label{subopt3:main}
	\begin{alignat}{3}
	&\mathop {\max }\limits_{\bm{\psi}} {\tilde G}\left( {\bm{\psi}} \right), \tag{\ref{subopt3:main}}\\
	\text{s.t.} \;  
	& {P_a^{min}} \le \frac {{c_3}{\varphi_a}{g_{c,a}}{\beta^{(\!n\!)}}}{\nu_a} \le {P_a^{max}}, \label{subopt3:a}\\ 
	& {P_h^{min}} \le \frac{{c_3}{\varphi_h}{g_{c,h}}{\left(\beta^{(\!n\!)} - \tau_h\right)}}{\mu_h} \le {P_h^{max}}, \label{subopt3:b}\\
	& {E_h^{min}} \le {{c_3}{\varphi_h}{g_{c,h}}{\left(\!\beta^{(\!n\!)} - \tau_h\!\right)}} \le {E_h^{max}}, \label{subopt3:c}\\
	& 0 \!\le\! \sum\nolimits_{p = 1}^P \!{\theta _p} \!+\! \sum\nolimits_{h = 1}^H \!{\tau _h} \!\le\! 1- \beta^{(\!n\!)},\forall {\theta _p}, {\tau _h} \!\ge\! 0, \label{subopt3:d}\\
	& 0 \!\le\! \sum\nolimits_{a = 1}^A\! {\nu _a}  \!+\! \sum\nolimits_{h = 1}^H \!{\mu _h} \!\le\! \beta^{(\!n\!)}, \forall {\nu _a}, {\mu _h} \!\ge\! 0,  \label{subopt3:e}
	\end{alignat}
\end{subequations}
where
\begin{equation}
\label{eq: funcG3}
\begin{aligned}
&{\tilde G}\!\left( {\bm{\psi}}  \right) \!=\!\! \sum\limits_{p = 1}^P\! {c_1} {\theta _p}{\log_2 \left(1 + {c_3}{\kappa_p}\right)} \!+\! \sum\limits_{a = 1}^A  \!{c_2}{\nu _a}{\log _2}\!\!\left(\!\! {1\! + \!\frac{{c_3}{\delta_a}{\beta^{(\!n\!)}}}{{{\nu _a}}}} \!\right) \\
&\quad \!\!+\!\! \sum\limits_{h = 1}^H\!\! {\left[\! {{c_1}{\tau _h}{\log_2 \!\left(1 \!+\! {c_3}{\kappa_h}\right)} \!+ \!{c_2}{\mu _h}{{\log }_2}\!\!\left(\! \!{1\! + \!\frac{{c_3}{\delta_h}{\left(\beta^{(\!n\!)} \!-\! \tau_h\!\right)}}{\mu _h}}\!\!\! \right)} \!\!\right]} \\
& \quad \! -{c_3}{p_l^{(n)}}{\beta^{(n)}}.
\end{aligned}
\end{equation}
It can be observed that the SNR constraints of backscatter devices, i.e., PWPDs and HWPDs, as well as the energy constraints for AWPDs and HWPDs are removed as they are always satisfied with the fixed $\{p_l^{(n)}, \beta^{(n)}\}.$

To obtain the optimal solution for the sub-problem \eqref{subopt3:main}, we have the following Lemma.
\begin{lemma}
	\label{lemma: convex_proof_subopt3}
	The objective function $G_{3}$ is a concave function w.r.t. $\bm{\psi}$ satisfying the linear constraints in \eqref{subopt3:a}-\eqref{subopt3:e}, and the optimal solution for the sub-problem~\eqref{subopt3:main} can be obtained by the interior-point method.
\end{lemma}

\begin{proof}
	See Appendix~\ref{App:lemma3}.
\end{proof}

\begin{algorithm}[t]
	\caption{The iterative algorithm for the PA scheme.}
	\label{algorithm1}
	\begin{algorithmic}[1]
		\State \textbf{Input:} The previous output $\{{p_l}^{(n-1)}, {\beta}^{(n-1)}, {\bm{\psi}}^{(n-1)}\}$.
		\State \textbf{Initialize:} $n = 1$, $\{{p_l}^{(0)}, {\beta}^{(0)}, {\bm{\psi}^{(0)}}\}$, tolerance $\xi_1 > 0$.
		\State \textbf{Compute:} the leader's utility ${\bm{U}_L}\left( {{{p_l}^{\left( 0 \right)}},{{\beta}^{\left( 0 \right)}}}, {\bm{\psi}}^{(0)} \right)$.
		\State \textbf{Repeat:}
		\State \quad {Obtain ${p_l}^{(n)}$ for given $\{{p_l}^{(n-1)}, {\beta}^{(n-1)}, {\bm{\psi}}^{(n-1)}\}$ by \hspace*{2.53mm} solving \eqref{subopt1:main};}
		\State \quad {Derive the optimal value ${\beta}^{(n)}$ with fixed \hspace*{2.53mm} $\{{p_l}^{(n)}, \bm{\psi}^{(n-1)}\}$ by solving \eqref{subopt2:main};}
		\State \quad {For given $\{{p_l}^{(\!n\!)}, {\beta}^{(\!n\!)}\},{\bm{\psi}}^{(\!n\!)}$ is obtained by solving \eqref{subopt3:main};}
		\State \quad {\textbf{If:}} 
		\State{\quad\quad $\left| {{\bm{U}_L}\!\left({{{p_l}^{\left(\!n\!\right)}}\!, {{\beta}^{\left(\!n\!\right)}}\!, {{\bm{\psi}}^{\left(\!n\!\right)}}}\!\right) \!-\! {\bm{U}_L}\!\left({{{p_l}^{\left(\!{n \!-\! 1}\!\right)}}\!, {{\beta}^{\left(\! {n \!-\! 1}\!\right)}}\!, {{\bm{\psi}}^{\left(\!n \!-\! 1 \right)}}}\! \right)}\!\right| \!<\! {\xi_1}$;} 
		\State \quad {\textbf{Then:}} 
		\State \quad\quad {Set $\{{\hat p_l}^*, \hat{\beta}^*, {\bm{\hat \psi}}^*\} \!=\! \{{p_l}^{(\!n\!)}, {\beta}^{(\!n\!)}, {\bm{\psi}}^{(\!n\!)}\}$ and terminate.}
		\State{\quad \textbf{Otherwise:}} 
		\State \quad \quad {Update $n \leftarrow  n + 1$} and continue.
		\State {\textbf{Output:}} The locally optimal solution $\bm{\hat \chi^*} = \{\hat {p_l}^*, \hat{\beta}^*, {\bm{\hat\psi}}^*\}$.
	\end{algorithmic}
\end{algorithm}

The overall proposed iterative algorithm is summarized in \textbf{Algorithm~\ref{algorithm1}}. The convergence of the proposed iterative algorithm is formally stated in the following theorem.
\begin{theorem}
	\label{theorem: convergence_complexity_BCD}
Algorithm~\ref{algorithm1} converges to a locally optimal solution of the Leader maximization problem.  
\end{theorem}

\begin{proof}
	See Appendix~\ref{App: BCD_proof}.
\end{proof}

\subsection{JA Scheme}
For the JA scheme, we first perform a joint optimization of the energy price and service time for the PB due to their trade-off relationship. The time scheduling for the IoT devices is then obtained given the optimal value of $\{{p_l}, {\beta}\}$ by solving the problem~\eqref{subopt3:main}.

\subsubsection{Joint Optimal Energy Price and Service Time}
With a given tuple $\{{p_l}^{(n-1)}, {\beta}^{(n-1)}, {\bm{\psi}}^{(n-1)}\}$ from the previous output, the optimal energy price and service time are the solutions of the following sub-problem:
\begin{subequations}\label{subopt4:main}
	\begin{alignat}{3}
	&\mathop {\max }\limits_{{p_l},{\beta}} {Q} \left( {p_l}, {\beta} \right), \tag{\ref{subopt4:main}}\\
	\text{s.t.} \quad & 0 \le {\beta} \le 1, \label{subopt4:a} \\
	& 0 \le {p_l} - {b_m} \le 2{a_m}{P_S^{\max }}, \label{subopt4:b} \\
	& {P_a^{min}} \le {\frac{{\varphi_a}{g_{c,a}}{\beta}}{2{\nu_a^{(\!n-1\!)}}{a_m}}{\left({p_l - b_m}\!\right)}} \le {P_a^{max}}, \label{subopt4:c} \\
	& {P_h^{min}} \!\le\! {\frac{{\varphi_h}{g_{c,h}}{\left(\!\beta - \tau_h^{(\!n-1\!)}\!\right)}}{2{\mu_h^{(\!n-1\!)}}\!{a_m}}{\left({p_l - b_m}\!\right)}} \!\le\! {P_h^{max}},\label{subopt4:d} \\
	& {E_a^{min}} \!\le \frac{{\varphi_a}{g_{c,a}}{\beta}}{2{a_m}}{\left({p_l \!-\! b_m}\right)} \!\le {E_a^{max}},\label{subopt4:e} \\
	& {E_h^{min}} \!\le {\frac{{\varphi_h}{g_{c,h}}{\left(\!\beta - \tau_h^{(\!n-1\!)}\!\right)}}{2{a_m}}{\left({p_l - b_m}\!\right)}} \!\le {E_h^{max}}, \label{subopt4:f} \\
	& {\kappa_a}\!\left( {{p_l} - {b_m}} \right) \ge {2{a_m}}{\gamma_p^{min}}, \label{subopt4:g}\\
	& {\kappa_h}\!\left( {{p_l} - {b_m}} \right) \ge {2{a_m}}{\gamma_h^{min}}, \label{subopt4:h}
	\end{alignat}	
\end{subequations}
where ${Q}\left( {{p_l},\beta } \right)$ is expressed in~\eqref{eq: func_G4}.
\begin{figure*}
\begin{equation}
\begin{aligned}
\label{eq: func_G4}
{Q}\left( {{p_l},\beta } \right) &\!=\! \sum\limits_{p = 1}^P \!{{c_1}{\theta_p^{(n-1)}}{{\log }_2}\! \left[ \! {1  \!+\! \frac{\kappa_p}{2{a_m}}\!\left( {{p_l} \!-\! {b_m}} \!\right)}\!  \right]}  \!+\!\! \sum\limits_{a = 1}^A \!{{c_2}{\nu_a^{(n-1)}}{{\log }_2}\! \left[  \!{1 \! +\!  \frac{\delta_a{\beta}}{2{\nu_a^{(n-1)}}{a_m}} \!\left( {{p_l} \!-\! {b_m}} \! \right)} \! \right]} \\ 
&\quad \!+\!  \sum\limits_{h = 1}^H \! \!\left\{ \! {{c_1}{\tau_h^{(n-1)}}{{\log }_2} \!\!\left[\!{1 \!+\! \frac{\kappa_p}{2{a_m}}\!\left( {{p_l} \!-\! {b_m}}\! \right)}  \!\right] } \!+\!{{c_2}{\mu_h^{(n-1)}}{{\log }_2} \!\!\left[ \! {1  \!+\!  \frac{{\delta_h}(\beta\!-\!\tau_h^{(n-1)})}{2{\mu_h^{(n-1)}}{a_m}}\!\left( {{p_l} \!-\! {b_m}} \!\right)}\!  \right]} \!\!\right\} - \frac{{{\beta {p_l}\!\left( {{p_l} \!-\! {b_m}}  \right)}}}{2{a_m}}, \quad \quad\quad
\end{aligned}
\end{equation}
\end{figure*}

However, the sub-problem~\eqref{subopt4:main} is also non-concave due to its non-convex feasible set, i.e., the constraints~\eqref{subopt4:c}-\eqref{subopt4:f} are non-convex w.r.t. $\left(p_l, \beta\right)$. To address this problem, we linearise the product ${\beta}{p_l}\left({p_l - b_m}\right)$ by defining $q_1 = \frac{1}{2}{\left({p_l - b_m}\right)}{\left(1 + \beta\right)}$, $q_2 = \frac{1}{2}\left({p_l - b_m}\right)\left(1 - \beta\right)$, then the problem~\eqref{subopt4:main} becomes:
\begin{subequations}
\label{subopt5:main}
\vspace{10pt}
\begin{alignat}{3}
	&\mathop {\max }\limits_{{q_1, q_2}} {\hat Q} \left( q_1, q_2 \right), \tag{\ref{subopt5:main}}\\
	\text{s.t.}  &  0 \le q_2 \le q_1, \label{subopt5:a} \\
	& 0 \le q_1 + q_2 \le 2{a_m}{P_S^{max}}, \label{subopt5:b} \\
	& {P_a^{min}} \!\le\! {\frac{{\varphi_a}{g_{c,a}}}{2{\nu_a^{(n-1)}}{a_m}}}{\left({q_1 \!-\! q_2}\right)} \!\le\! {P_a^{max}}\!,  \label{subopt5:c} \\
	& {P_h^{min}} \!\!\le\!\! \frac{{\varphi_h}{g_{c,h}}\!{\left[\!\!\left(\!\!1 \!-\! \tau_h^{(\!n-1\!)}\!\!\right)\!\!{q_1} \!-\! \left(\!\!1 \!+\! \tau_h^{(\!n-1\!)}\!\!\right)\!\!{q_2}\!\right]}}{2{\mu_h^{(n-1)}}{a_m}} \!\!\le\! {P_h^{max}}\!, \label{subopt5:d} \\
	& {E_a^{min}} \!\le\! {\frac{{\varphi_a}{g_{c,a}}}{2{a_m}}}{\left({q_1 \!-\! q_2}\right)} \!\le\! {E_a^{max}}\!, \label{subopt5:e} \\
	& {E_h^{min}} \!\le\!\! \frac{{\varphi_h}{g_{c,h}}\!{\left[\!\!\left(\!\!1 \!-\! \tau_h^{(\!n-1\!)}\!\!\right)\!\!{q_1} \!-\! \left(\!\!1 \!+\! \tau_h^{(\!n-1\!)}\!\!\right)\!\!{q_2}\!\right]}}{2{a_m}} \!\!\le\! {E_h^{max}}\!, \label{subopt5:f} \\
	& {\kappa_p}\left( {q_1 + q_2} \right) \ge {2{a_m}}{\gamma_p^{min}}, \label{subopt5:g}\\
	& {\kappa_h}\left( {q_1 + q_2} \right) \ge {2{a_m}}{\gamma_h^{min}}, \label{subopt5:h}
\end{alignat}	
\end{subequations}
where ${\hat Q}\!\left( {{q_1},{q_2}} \right)$ is expressed in~\eqref{eq: func_Q}.
\begin{figure*}
\begin{equation}
\begin{aligned}
\label{eq: func_Q}
&{\hat Q}\!\left( {{q_1},{q_2}} \right) \!=\!\! \sum\limits_{p = 1}^P {{c_1}{\theta_p^{(n-1)}}{{\log }_2}\left[ {1 + \frac{\kappa_p{\left( {{q_1} + {q_2}} \right)}}{2{a_m}}} \right]} \!+ \!\!\sum\limits_{a = 1}^A \!{{c_2}{\nu_a^{(\!n-1\!)}}{{\log }_2}\!\!\left[\!{1 \!+\! \frac{{\delta_a}\left( {{q_1} \!-\! {q_2}} \right)}{2{\nu_a^{(\!n-1\!)}}{a_m}}} \!\right]} \\
&\quad\!+\!\! \sum\limits_{h = 1}^H \!\!\left\{\!\! {{c_1}{\tau_h^{(\!n-1\!)}}{{\log }_2}\!\!\left[ \!{1 \!+\! \frac{{\kappa_h}\!\left(\!{{q_1} \!+\! {q_2}} \!\right)}{2{a_m}}}\! \right]} \!\!+\! {{c_2}{\mu_h^{(\!n-1\!)}}{{\log }_2}\!\!\left\{\! \!{1 \!+\! \frac{{\delta_h}\!\!\left[\!(\!1 \!-\! \tau_h^{(\!n-1\!) }){q_1} \!-\! (\!1 \!+\! \tau_h^{\left(\!n-1\!\right)}){q_2}\right]}{2{\mu_h^{(n-1)}}{a_m}}} \!\! \right\}} \!\!\right\} \!-\! \frac{\left(\!{q_1^2 \!+\! {b_m}{q_1}}\!\right)}{2{a_m}} \!+\! \frac{\left(\!{q_2^2 \!+\! {b_m}{q_2}}\!\right)}{2{a_m}},
\end{aligned}
\end{equation}
\end{figure*}
Intuitively, the last term of the objective function ${\hat Q}{(q_1, q_2)}$ is convex, while other terms are concave.  We define ${V} \buildrel \Delta \over =  \{q_1, q_2\}$ and $S_V$ to be the set of $V$ satisfying~\eqref{subopt5:a}-\eqref{subopt5:h}, then the objective function of the problem~\eqref{subopt5:main} is rewritten as follows:
\begin{equation}
	{\hat Q}\left( V \right) = {Q_{ccav}}\left( V \right) + {Q_{cvex}}\left( V \right),
\end{equation}
where $Q_{ccav}\!\left( V \right)$ is expressed in~\eqref{eq: Q_ccav}, and ${Q_{cvex}} \!\left( V \right)\!=\! \frac{\left( {q_2^2 + {b_m}{q_2}} \right)}{2{a_m}}$.
\begin{figure*}[!]
\begin{equation}
\label{eq: Q_ccav}
\begin{aligned}
 {Q_{ccav}}\!\left(\!V\!\right) &= \sum\limits_{p = 1}^P {{c_1}{\theta_p^{(n-1)}}{{\log }_2}\left[ {1 + \frac{\kappa_p{\left( {{q_1} + {q_2}} \right)}}{2{a_m}}} \right]} \!+ \!\!\sum\limits_{a = 1}^A \!{{c_2}{\nu_a^{(\!n-1\!)}}{{\log }_2}\!\!\left[\!{1 \!+\! \frac{{\delta_a}\left( {{q_1} \!-\! {q_2}} \right)}{2{\nu_a^{(\!n-1\!)}}{a_m}}} \!\right]} \\
 &\quad\!+\!\! \sum\limits_{h = 1}^H \!\!\left\{\!\! {{c_1}{\tau_h^{(\!n-1\!)}}{{\log }_2}\!\!\left[ \!{1 \!+\! \frac{{\kappa_h}\!\left(\!{{q_1} \!+\! {q_2}} \!\right)}{2{a_m}}}\! \right]} \!\!+\! {{c_2}{\mu_h^{(\!n-1\!)}}{{\log }_2}\!\!\left\{\! \!{1 \!+\! \frac{{\delta_h}\!\!\left[\!(\!1 \!-\! \tau_h^{(\!n-1\!) }){q_1} \!-\! (\!1 \!+\! \tau_h^{\left(\!n-1\!\right)}){q_2}\right]}{2{\mu_h^{(n-1)}}{a_m}}} \!\! \right\}} \!\!\right\} \!-\! \frac{\left(\!{q_1^2 \!+\! {b_m}{q_1}}\!\right)}{2{a_m}}, \quad\quad
\end{aligned}
\end{equation}
	\hrulefill
\end{figure*}
\begin{algorithm}[t]
	\caption{The CCCP algorithm to solve the DC programming problem in \eqref{subopt5:main}.}
	\label{algorithm2}
	\begin{algorithmic}[1]
		\State \textbf{Input:} The previous result of the BCD algorithm $\{{p_l^{(n-1)}, \beta^{(n-1)}, \bm{\psi}^{(n-1)}}\}$.
		\State \textbf{Initialize:} Initiate $k \!=\! 1$, a tolerance $\xi_2 \!>\! 0$, and a feasible solution $V^{(0)} \!\!=\!\! \{q_1^{(0)}(p_l^{(\!n-1\!)}\!, \beta^{(n-1)}), q_2^{(0)}(p_l^{(n-1)}\!, \beta^{(n-1)}\!)\!\}$.
		\State \textbf{Repeat:}
		\State \quad {Transform~\eqref{subopt5:main} into~\eqref{eq: subopt6};}
		\State \quad {Obtain the optimal $V^{(k)}$ by solving~\eqref{eq: subopt6};} 
		\State \quad {\textbf{If:}} 
		\State{\quad\quad $ \left|  {\hat Q}\left( {{V^{\left( k \right)}}} \right) - {\hat Q}\left( {{V^{\left( {k - 1} \right)}}} \right) \right| < {\xi_2}$;} 
		\State \quad {\textbf{Then:}} 
		\State \quad\quad {Set $ V^* = V^{(k)}$ and terminate.}
		\State{\quad \textbf{Otherwise:}} 
		\State \quad \quad {Update $k \leftarrow  k + 1$} and continue.
		\State {\textbf{Output:}} The locally optimal solution $ V^* = \{{q_1}^*, {q_2}^*\}$.
	\end{algorithmic}
\end{algorithm}
The problem~\eqref{subopt5:main} is the \textit{difference-of-convex-function} (DC) programming problem, which can be solved efficiently by the \textit{convex-concave procedure} (CCCP)~\cite{Yuille2001}. The key idea of the CCCP is to linearise the last term (i.e., a convex function) by the first-order Taylor expansion at the current fixed point. We denote ${V^{(\!k-1)}} \!\buildrel \Delta \over =  \left\{\!q_1^{(\!k-1\!)}\!, q_2^{(\!k-1\!)}\!\!\right\}$ as the fixed point at the $k$-th iteration, then the problem~\eqref{subopt5:main} can be solved by the following sequential convex programming with linear constraints~\eqref{subopt5:a}-\eqref{subopt5:h}:
\begin{equation}
	\begin{aligned}
	\label{eq: subopt6}
	{V^{\left( k \right)}} & \buildrel \Delta \over =  \argmax \limits_{V \in S} {\tilde Q}\left( V \right) \\
	&=  \argmax\limits_{V \in S} \left\{ {{Q_{ccav}}\!\left( V \right) + {V^T}\nabla {Q_{cvex}}\!\left( \!{{V^{\left( {k - 1} \right)}}} \right)}\! \!\right\},
	\end{aligned}	
\end{equation}
where $\nabla {Q_{cvex}}\left( {{V^{\left( {k - 1} \right)}}} \right) = \frac{(2q_2^{\left( {k - 1} \right)} + {b_m})}{2{a_m}}$ is the gradient of $Q_{cvex}\left(V\right)$ at $V^{(k-1)}$. Ultimately, ${\tilde Q} (V)$ is a convex function, thus $V^{(k)}$ can be easily obtained by numerical methods such as the Newton or Interior-point methods. 

In general, the CCCP can start at any point within the feasible region defined by the constraints~\eqref{subopt5:a}-\eqref{subopt5:h} when it stands alone. However, we choose the initial value $V^{(0)}$ to guarantee the convergence of the outer iterative algorithm (i.e., the BCD algorithm) as follows:
\begin{equation}
	V^{(0)} \!\!=\!\! \left\{\!q_1^{(0)}\!\!\left(\!p_l^{(\!n-1\!)}\!, \beta^{(\!n-1\!)}\!\right)\!\!, q_2^{(0)}\!\!\left(\!p_l^{(\!n-1\!)}\!, \beta^{(\!n-1\!)}\!\right)\!\!\right\}.
\end{equation}
The entire procedure of the CCCP algorithm is summarized in \textbf{Algorithm~\ref{algorithm2}}.
The convergence of the CCCP algorithm is formally stated in the following Theorem:

\begin{theorem}
	\label{theorem: convergence_optimal_solution_CCCP}
Algorithm~\ref{algorithm2} (utilizing the CCCP technique to solve the joint optimization problem~\eqref{subopt5:main}) converges to a local optimum $V^{*}$ by generating a sequence of $V^{(k)}$ providing ${\hat Q}\left( {{V^{\left( k \right)}}} \right) > {\hat Q}\left( {{V^{\left( {k - 1} \right)}}} \right), \forall k \ge 1$.  
\end{theorem}
\begin{proof}
	See Appendix~\ref{App: CCCP_proof}.
\end{proof}
\subsubsection{The Overall Iterative Algorithm for JA Scheme}

After the implementation of the joint energy price and service time estimation, we perform time allocation for the IoT devices optimally by solving the problem in~\eqref{subopt3:main}. These steps are repeated until the stopping criterion of the algorithm is satisfied. The overall iterative algorithm for the JA scheme is summarized in \textbf{Algorithm~\ref{algorithm3}}.
\begin{algorithm}[t]
	\caption{The iterative algorithm for the JA scheme.}
	\label{algorithm3}
	\begin{algorithmic}[1]
		\State \textbf{Input:} The previous output $\{{p_l}^{(n-1)}, {\beta}^{(n-1)}, {\bm{\psi}}^{(n-1)}\}$.
		\State \textbf{Initialize:} $n = 1$, $\{{p_l}^{(0)}, {\beta}^{(0)}, {\bm{\psi}^{(0)}}\}$, tolerance $\xi_1 > 0$.
		\State \textbf{Compute:} the leader's utility ${\bm{U}_L}\!\left( {{{p_l}^{\left( 0 \right)}},{{\beta}^{\left( 0 \right)}}}, {\bm{\psi}}^{(0)} \right)$.
		\State \textbf{Repeat:}
		\State \quad {Obtain the joint optimal $\{{p_l}^{(\!n\!)}\!, {\beta^{(\!n\!)}}\!\}$ from the previous \hspace*{2.53mm} output $\{{p_l}^{(\!n-1\!)}\!, {\beta}^{(\!n-1\!)}\!, {\bm{\psi}}^{(\!n-1\!)}\!\}$ by processing the CCCP \hspace*{2.53mm} algorithm to solve the problem~\eqref{subopt5:main};}
		\State \quad {Derive the optimal ${\bm{\psi}}^{(\!n\!)}\!$ with fixed $\{{p_l}^{(\!n\!)}\!, \!\beta^{(n)}\!\}$ by solving \hspace*{2.53mm} the problem~\eqref{subopt2:main};}
		\State \quad {\textbf{If:}} 
		\State{\quad\quad $\left| {{\bm{U}_L}\!\left({{{p_l}^{\left(\!n\!\right)}}\!, {{\beta}^{\left(\!n\!\right)}}\!, {{\bm{\psi}}^{\left(\!n\!\right)}}}\!\right) \!-\! {\bm{U}_L}\!\left({{{p_l}^{\left(\!{n \!-\! 1}\!\right)}}\!, {{\beta}^{\left(\! {n \!-\! 1}\!\right)}}\!, {{\bm{\psi}}^{\left(\!n \!-\! 1 \right)}}}\! \right)}\!\right| \!<\! {\xi_1}$;}  
		\State \quad {\textbf{Then:}} 
		\State \quad\quad {Set $\{\hat {p_l}^*\!, \hat {\beta}^*\!, {\bm{\hat\psi}}^*\!\} = \{{p_l}^{(\!n\!)}\!, {\beta}^{(\!n\!)}\!, {\bm{\psi}}^{(\!n\!)}\!\}$ and terminate.}
		\State{\quad \textbf{Otherwise:}} 
		\State \quad \quad {Update $n \leftarrow  n + 1$} and continue.
		\State {\textbf{Output:}} The locally optimal solution $\bm{\hat \chi^*} = \{\hat{p_l}^*, \hat{\beta}^*, {\bm{\hat\psi}}^*\}$.
	\end{algorithmic}
\end{algorithm}

\begin{theorem}
	\label{theorem: convergence_optimal_solution_BCD2}
Algorithm~\ref{algorithm3} converges to a locally optimal solution of the Leader maximization problem.  
\end{theorem}
\begin{proof}
	Similar to the proof of the Theorem~\ref{theorem: convergence_complexity_BCD}.
\end{proof}

Finally, we can obtain a local SE for the proposed Stackelberg game, formally stated in the following Theorem.
\begin{theorem}
	\label{theorem: localSE}
	A local optimum ${\bm{\hat \chi}^*}$ obtained by the Theorem~\ref{theorem: convergence_complexity_BCD} or Theorem~\ref{theorem: convergence_optimal_solution_BCD2} combined with the optimal $P_S^*$ of the follower constitutes a local SE satisfying the Definition~\ref{definition:localSE}.
\end{theorem}
\begin{proof}
	It is worth noting that the output $\bm{\hat{\chi}}^*$ obtained by the Theorem~\ref{theorem: convergence_complexity_BCD} or Theorem~\ref{theorem: convergence_optimal_solution_BCD2} is a locally optimal solution of the problem~\eqref{opt1:main}. It means that ${\bm{U}_L}\!{\left({\bm{\hat \chi}^*}, {P_S^*}\right)} \!\!\ge\!\! {\bm{U}_L}{\left({\bm{ \chi}}, {P_S^*}\right)}$ in the neighborhood of $\bm{\hat \chi}^*$. 
	Hence, $\bm{\hat \chi}^*$ is also a local SE of the leader in $\hat S_{\bm \chi} \subset S_{\bm \chi}$ that satisfying the Definition~\ref{definition:localSE}. It combines with the optimal solution $P_S^*$ to constitute the local SE for the proposed Stackelberg game.
\end{proof}

\section{Efficiency of the local Stackelberg Equilibrium}
\label{section:Efficiency_SE}
Energy trading based on Stackelberg game formulated in Section~\ref{section:EnergyTrading} captures the strategic interaction between the ISP and ESP. The optimal trading strategies of the ISP and ESP just aim to selfishly maximize each player's own profit. These strategies hence may lead to the performance loss in terms of the total profit achieved by both the ISP and ESP (often referred to as the social welfare). To evaluate the efficiency of the achieved local \textit{Stackelberg equilibrium} (SE), we introduce a \textit{social welfare} maximization approach, as a baseline scenario.

\subsection{Socially Optimal Welfare Scenario}
In the socially optimal welfare scenario, the ISP and ESP cooperatively maximize the sum of their profits. Mathematically, the utility function of social welfare can be formulated as
\begin{equation}
\begin{aligned}
\label{eq: social_welfare_func}
{{\bm{U}_{SW}}}\! \left(\! {{P_S}}, \beta, \bm{\psi} \right)  =  {\bm{U}_T}\!\left(\! {{P_S},\beta ,\psi } \right) \!-\! {\beta}\!\left({a_m}{P_S^2} \!+\! {b_m}{P_S}\right)\!, 
\end{aligned}
\end{equation}
\begin{figure*}[t]
	\begin{equation}
	\begin{aligned}
	\label{eq: profit_func}
	{\bm{U}_T}\!\left(\! {{P_S}\!,\beta \!,\bm{\psi}\! } \right) \!=\!  {p_r}\!\left\{ \sum\limits_{p = 1}^P\! {{\Omega_B}{\theta_p}{\log _2}\!\left(\!{1 \!+\! {\kappa _p}{P_S}} \!\right)} \!+\!\!\sum\limits_{a = 1}^A\! {\Omega_D}{{\nu _a}} {{\log }_2}\!\!\left(\! {\!1 \!+\! {\delta _a}\!\frac{{\beta{P_S}}}{{{\nu _a}}}} \!\!\right) \!\!+\!\! \sum\limits_{h = 1}^H \!\! \left[ \!{{{\Omega_B}{\tau_h}{\log _2}\!\left(\!{1 \!+\! {\kappa _h}{P_S} } \!\right)} \!\!+\!\! {\Omega_D}{\mu _h}{{\log }_2} \!\!\left(\!\! {1 \!+\! {\delta _h}\!\frac{{( \beta  \!-\! {\tau _h}){P_S}}}{{{\mu _h}}}} \!\!\right)} \!\!\right]\!\! \!\right\}\!.
	\end{aligned}
	\end{equation}
	\hrulefill
\end{figure*}
\begin{figure*}[!]
	\begin{center}
		$\begin{array}{ccc} 
		\epsfxsize= 2.8 in \epsffile{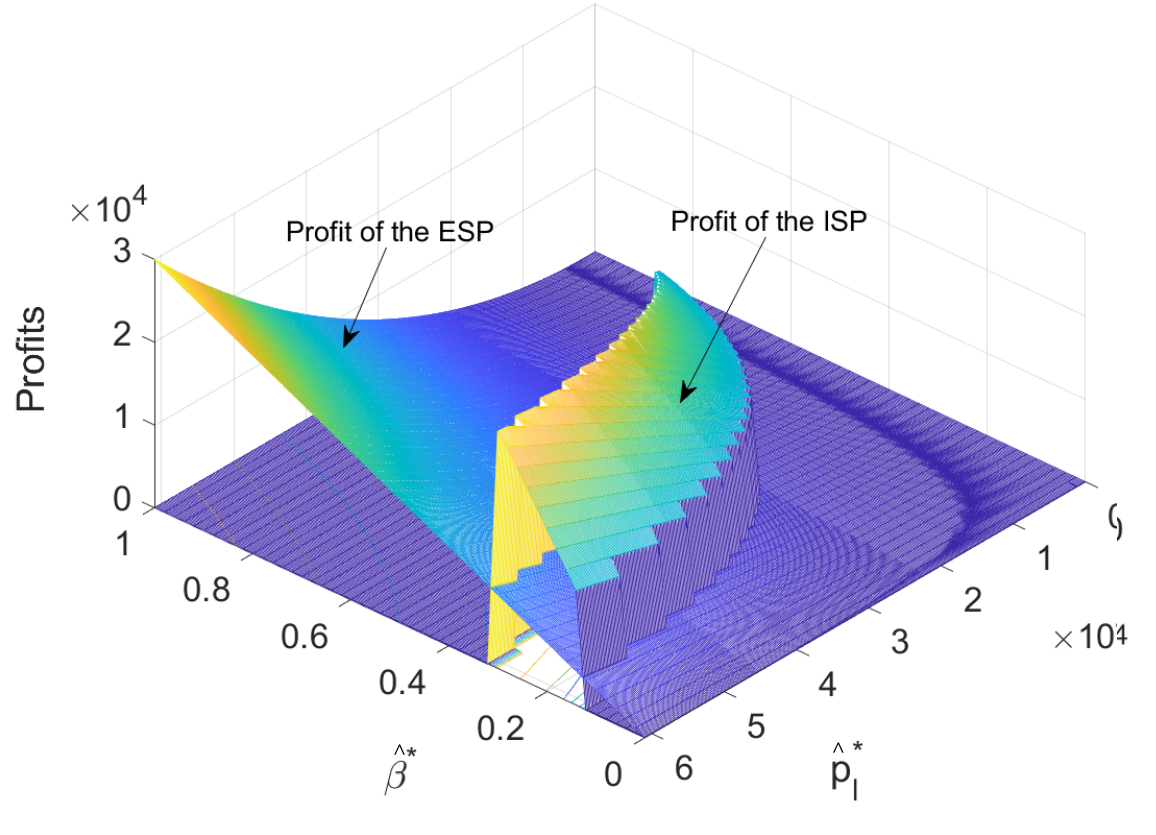}  \quad & 
		\epsfxsize= 2.5 in \epsffile{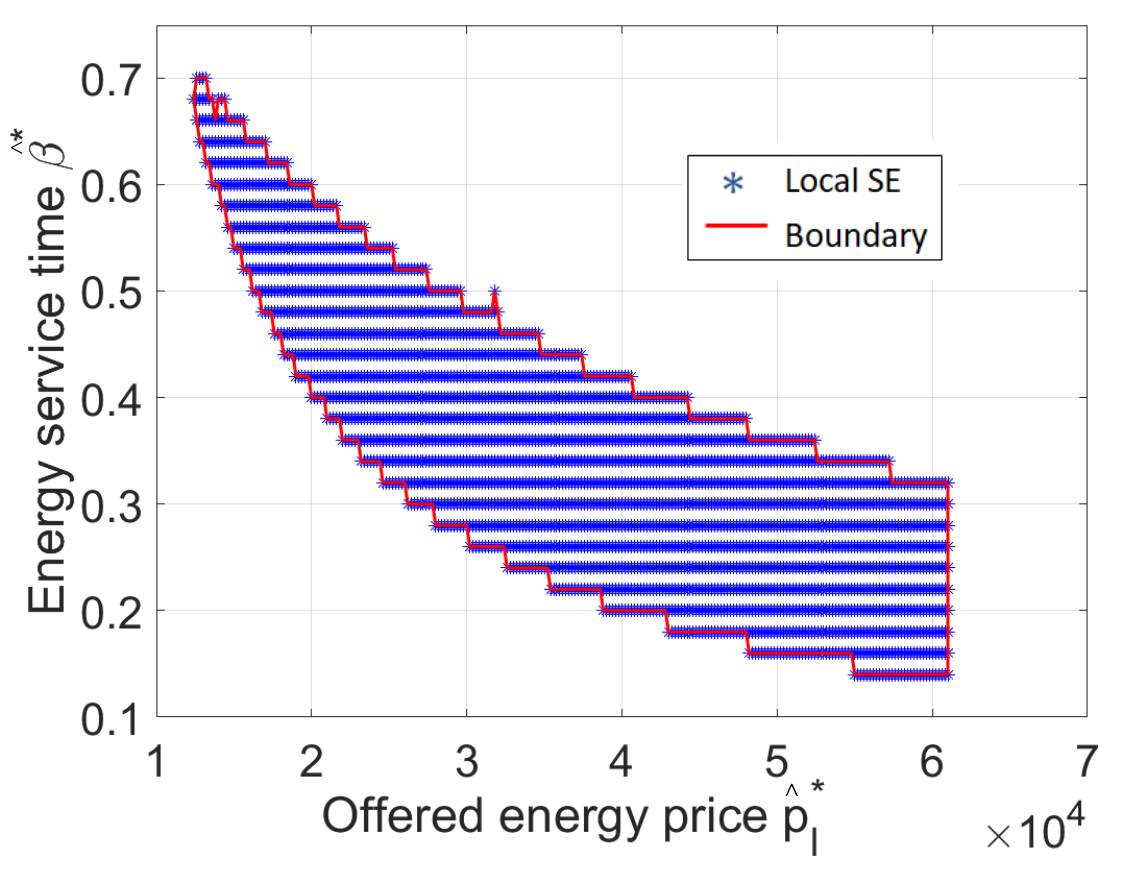} \\ [0.2cm]
		(a) & \quad (b)
		\end{array}$
		\caption{(a) Profits of the leader and follower, (b) Different local SEs vs. locally optimal offered price and energy service time.}
		\label{fig: FollowerUtility}
	\end{center}
\end{figure*}
where $\bm{U}_T$$\left(\! {{P_S}\!,\beta \!,\bm{\psi}\! } \right)$ is in~\eqref{eq: profit_func}.
The social welfare maximization problem is given by:
\begin{subequations}\label{optSW:main}
	\begin{align}
	&\mathop {\max }\limits_{\left( {{P_S}, \beta, \bm{\psi} } \right)}{{\bm{U}}_{SW}} \!\left( {{P_S}, \beta}, \bm{\psi} \right), \tag{\ref{optSW:main}}\\
	\text{s.t.} \; & 0 \le P_S \le P_{S}^{max}, \label{optSW:a} \\
	& P_a^{min} \le \frac{{{\varphi _a}{g_{c,a}}\beta {P_S}}}{{\nu_a}} \le P_a^{max}, \label{optSW:b}\\
	& P_h^{min} \le \frac{{{\varphi _h}{g_{c,h}}\left({\beta-{\tau _h}}\right){P_S}}}{{\mu_h}} \le P_h^{max}, \label{optSW:c}\\
	& E_a^{min} \le {{{\varphi _a}{g_{c,a}}\beta {P_S}}} \le {E_a^{max}}, \label{optSW:d}\\
	& E_h^{min} \le {{{\varphi _h}{g_{c,h}}\left({\beta-{\tau _h}}\right){P_S}}} \le {E_h^{max}}, \label{optSW:e}\\
	& {\kappa_p}{P_S} \ge \gamma_p^{\min }, \label{optSW:f}\\
	& {\kappa_h}{P_S} \ge \gamma_h^{\min }, \label{optSW:g}\\
	& 0 \le \sum\nolimits_{p = 1}^P {\theta _p} \!+\! \sum\nolimits_{h = 1}^H {\tau _h} \le \beta \le 1, \forall {\theta _p}, \forall {\tau _h} \!\ge\! 0, \label{optSW:h}\\
	&  0 \!\le\!\! \sum\nolimits_{a = 1}^A\! {\nu _a}  \!+\!\! \sum\nolimits_{h = 1}^H \!{\mu _h} \!\le\! 1 \!-\! \beta \!\le\! 1, \forall {\nu _a}, \forall {\mu _h} \!\ge\! 0.  \label{optSW:i}
	\end{align}
\end{subequations}
Similar to the problem~\eqref{opt1:main}, the social welfare maximization problem is also a non-concave problem due to the non-convexity of the constraints~\eqref{optSW:d} and~\eqref{optSW:e}. It can be also solved efficiently for locally optimal solutions with the partitioning schemes proposed in Section~\ref{section:IterativeAlgorithms}.

\subsection{Price of Anarchy}
To quantify the efficiency of the local SE of the proposed non-cooperative game, we uses the \textit{Price of Anarchy} (PoA)~\cite{Roughgarden2015} which is defined as the ratio of the utility value (i.e., formulated in~\eqref{eq: social_welfare_func}) at the worst local SE \cite{Lee2015Distributed} to its maximum value:
\begin{equation}
PoA = \frac{{{\bm{U}_{SW}}\left( {\bm{{\hat \chi ^*}}} \right)}}{{\mathop {\max }\limits_{\left( {{P_S},\beta ,{\rm{ }}\bm{\psi} } \right)} {\bm{U}_{SW}}\left( {{P_S},\beta ,\bm{\psi} } \right)}}.
\end{equation}

\section{Numerical Results}
\label{section:NumericalResults}
In this section, we first investigate the existence of the local \textit{Stackelberg equilibrium} (SE). We then evaluate and compare the proposed framework with ones that are designed for conventional transmission modes. Last but not least, we investigate the efficiency of the local SE. The carrier frequency of RF signals is set at $2.4$ GHz. The bandwidth of the RF signals and the antenna gain of the PB are $10$ MHz and $6$ dBi, respectively. The AWPDs and HWPDs have the signal bandwidth of $1$ MHz and antenna gains of $1.8$ dBi~\cite{Lyu2018Throughput}. Unless otherwise specified, the default distances between the PB and IoT devices are 10 meters (m) and the number of IoT devices is 10. In our setup, the AWPDs and HWPDs have the energy harvesting efficiency coefficients of $\varphi = 0.6$, whilst the scattering efficiency $\eta$ causes a power loss of 1.1 dB at the PWPDs and HWPDs~\cite{Lyu2018Throughput}. In addition, performance gap and noise psd of IoT devices are set at $\zeta = -5$ dB and $N_0 = -100$ dBm, respectively~\cite{Kim2017Hybrid}. A Dell computer with a CPU Intel Core i7-8565U, 16 GB RAM, and GPU Radeon RX 540 series, running MATLAB, is used for our simulations.
\begin{figure*}[!]
	\begin{center}
		$\begin{array}{ccc} 
		\epsfxsize= 2.8 in \epsffile{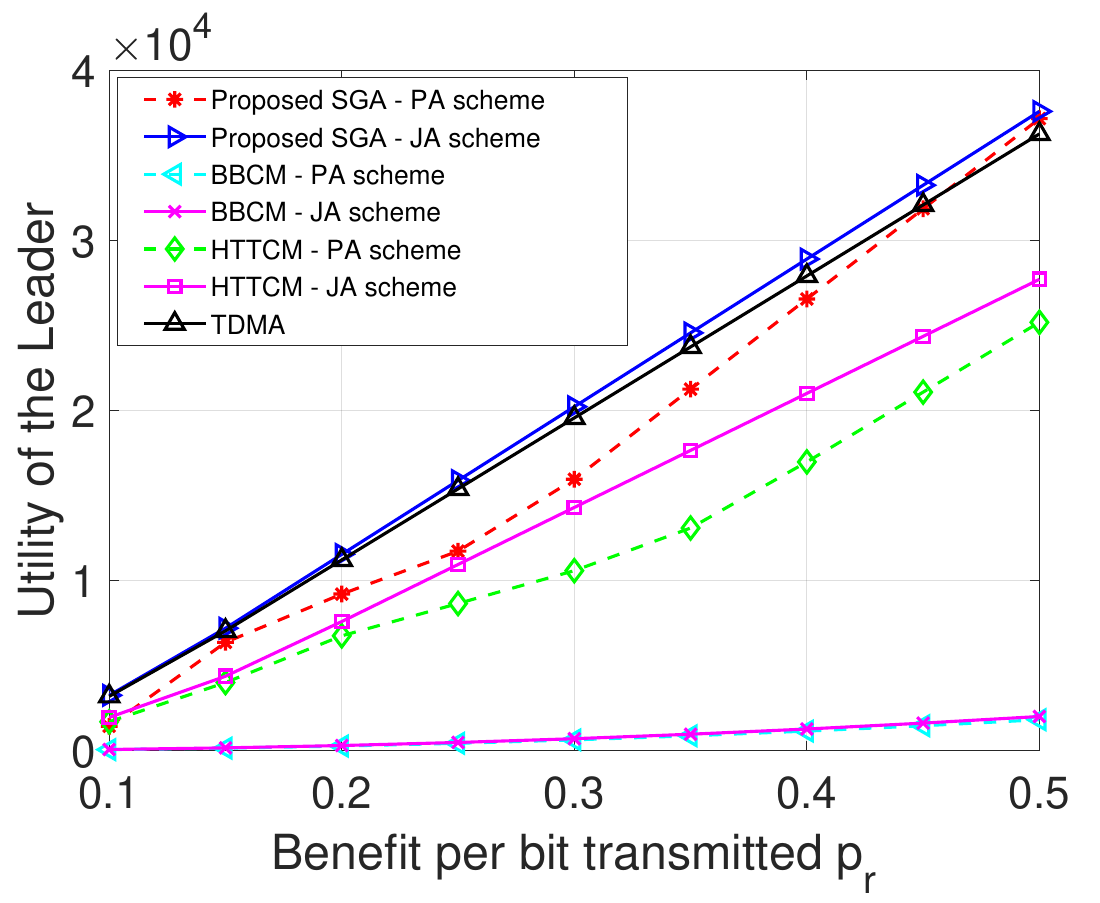}  \quad & 
		\epsfxsize= 2.9 in \epsffile{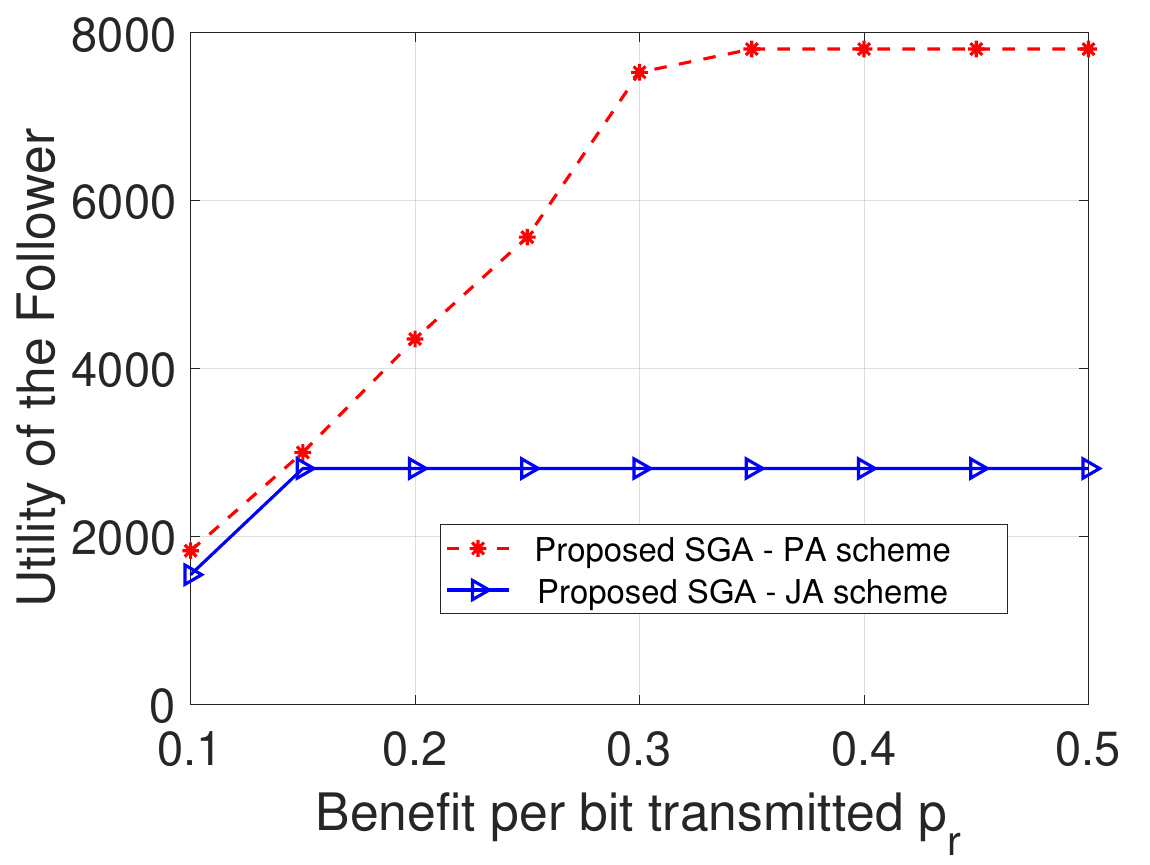} \\ [0.2cm]
		(a) & \quad (b)
		\end{array}$
		\caption{(a) Leader's payoff and (b) Follower's payoff vs. benefit per bit transmitted.}
		\label{fig:Varying_pr}
	\end{center}
\end{figure*}
\subsection{Existence of the Local Stackelberg Equilibrium}
In this subsection, we investigate the existence of the local SE in Fig.~\ref{fig: FollowerUtility} with respect to the locally optimal offered price $\hat p_l^*$ and energy service time $\hat \beta^*$. With the given $\left(\hat p_l^*, \hat \beta^*\right)$ from the ISP, we can obtain the maximal profit of the ESP by finding the unique optimal value of its transmission power $P_S^*$. Intuitively, this profit increases linearly with the $\hat \beta^*$ and non-linearly with the $\hat p_l^*$ as demonstrated in Fig.~\ref{fig: FollowerUtility}(a). On the other hand, the ISP can obtain different locally optimal solutions $\left(\hat p_l^*, \hat \beta^*\right)$ by using the PA/JA schemes which depend on the different initial values of the $p_l$ and $\beta$. Thus, in our proposed game, we can obtain different local SEs which are shown in Fig.~\ref{fig: FollowerUtility}(b). It can be also seen when the locally optimal offered energy price is high, the ISP prefers to choose a short energy service time and vice versa.

\subsection{Profit of ISP and ESP}
For profit comparison, we consider three conventional communication methods, i.e., the BBCM, HTTCM, and TDMA mechanism, which are implemented at the ISP. It is worth noting that for the TDMA mechanism, all IoT devices are allocated with identical time resources. In this case, the total backscatter time of the IoT devices accounts for a half of the normalized time frame as illustrated in Fig.~\ref{fig:time_frame}. Thus, the operation time of the PB $\beta$ is fixed and equal to the total backscatter time of the IoT devices.

\subsubsection{Impact of benefit per bit transmitted}
Fig.~\ref{fig:Varying_pr}(a) shows the utility of the leader (i.e., the ISP) and the follower (i.e., the ESP) when the benefit per bit transmitted $p_r$  in the range of $0.1$ to $0.5$ price unit. Obviously, the profits of the ISP obtained by all approaches increase when the benefit per bit transmitted increases. 
In particular, we first observe that the proposed Stackelberg game approach (SGA), BBCM, and HTTCM solved by the JA scheme always perform better than themselves solved by the PA scheme. The reason is that the JA scheme can optimize the profit of the ISP with respect to both the offered price $p_l$ and the active time of the PB $\beta$. In addition, we also observe that the proposed SGA solved by the JA scheme achieves the highest profit in the considered range of $p_r$. By contrast, the proposed SGA solved by the PA scheme obtains a lower profit for the ISP than that obtained by the TDMA mechanism when the benefit per bit transmitted is smaller than 0.5. This is because the PA scheme tends to offer a high price to get a high transmission power rather than choosing a long period for energy purchasing. Thus, the optimal energy purchasing time in the PA scheme is smaller than that in the TDMA mechanism. In addition, the offered price has more weight than the energy service time in the energy cost. As a result, the PA scheme may perform not as good as the TDMA mechanism in terms of the ISP's profit. Furthermore, due to the low backscatter efficiency, the BBCM solved by both schemes performs much worse than other methods in terms of the ISP's profit, in which the one solved by the JA scheme is slightly better than the one solved by the PA scheme.

The achieved profits of the ESP corresponding to those of the ISP which are optimized by the PA and JA schemes in the proposed game are in Fig.~\ref{fig:Varying_pr}(b). It can be observed that the profit of the ESP in the case using the PA scheme is higher than that using the JA scheme. As mentioned above, the reason is that the ISP using the PA scheme buys a higher transmission power than that using the JA scheme with a given benefit per bit transmitted. 

\subsubsection{Impact of Number of IoT Devices}
\begin{figure*}[!]
	\begin{center}
		$\begin{array}{ccc} 
		\epsfxsize=2.2 in \epsffile{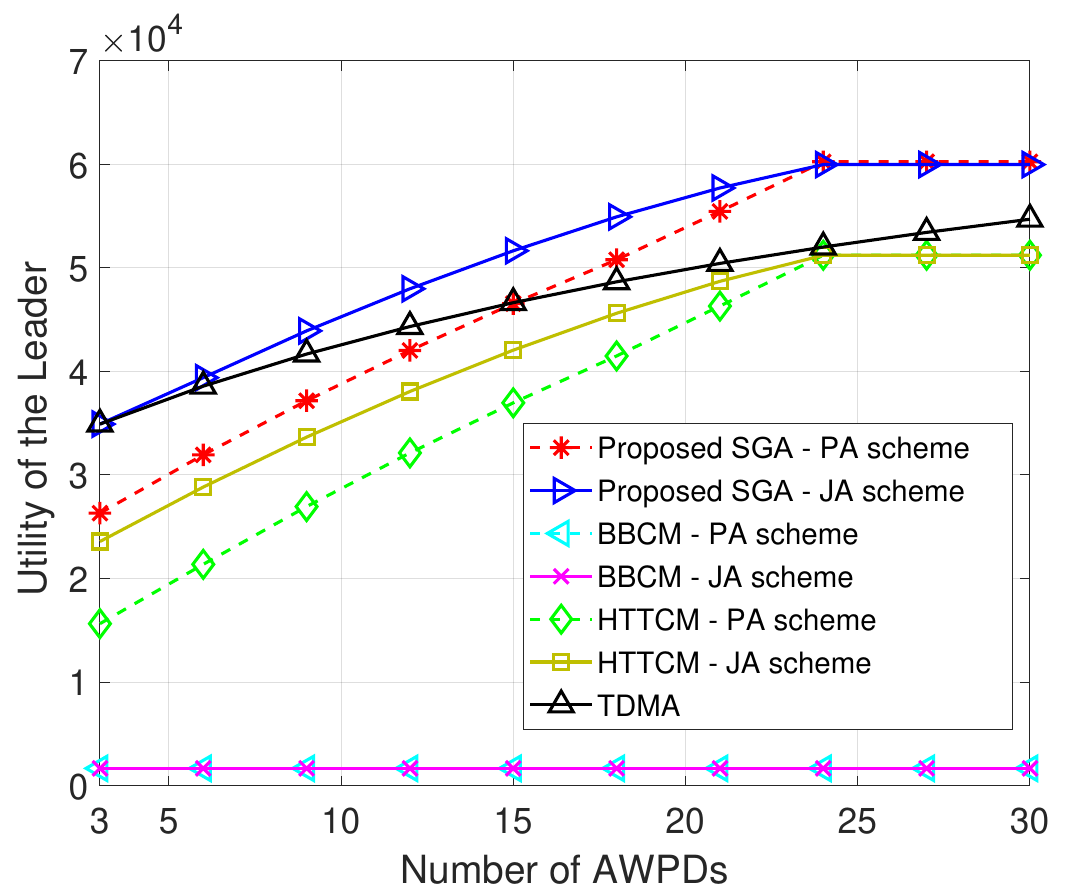} &
		\epsfxsize=2.2 in \epsffile{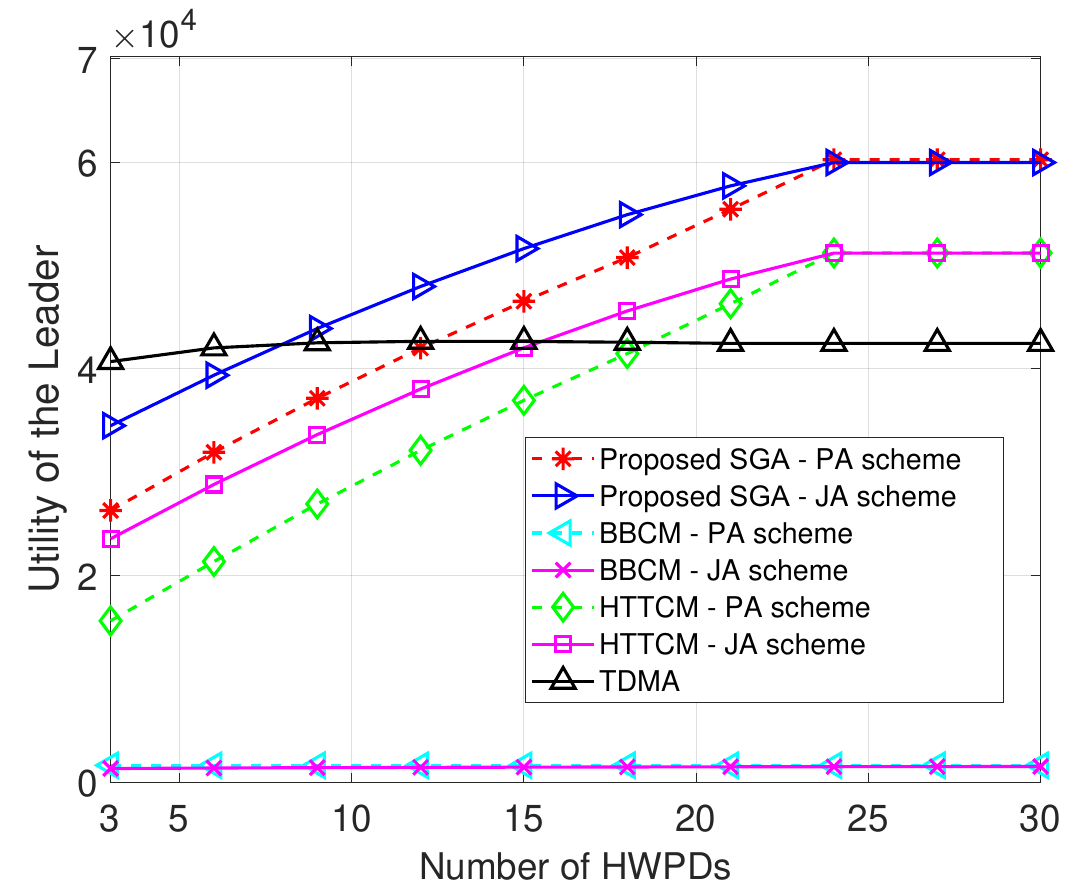} &
		\epsfxsize=2.2 in \epsffile{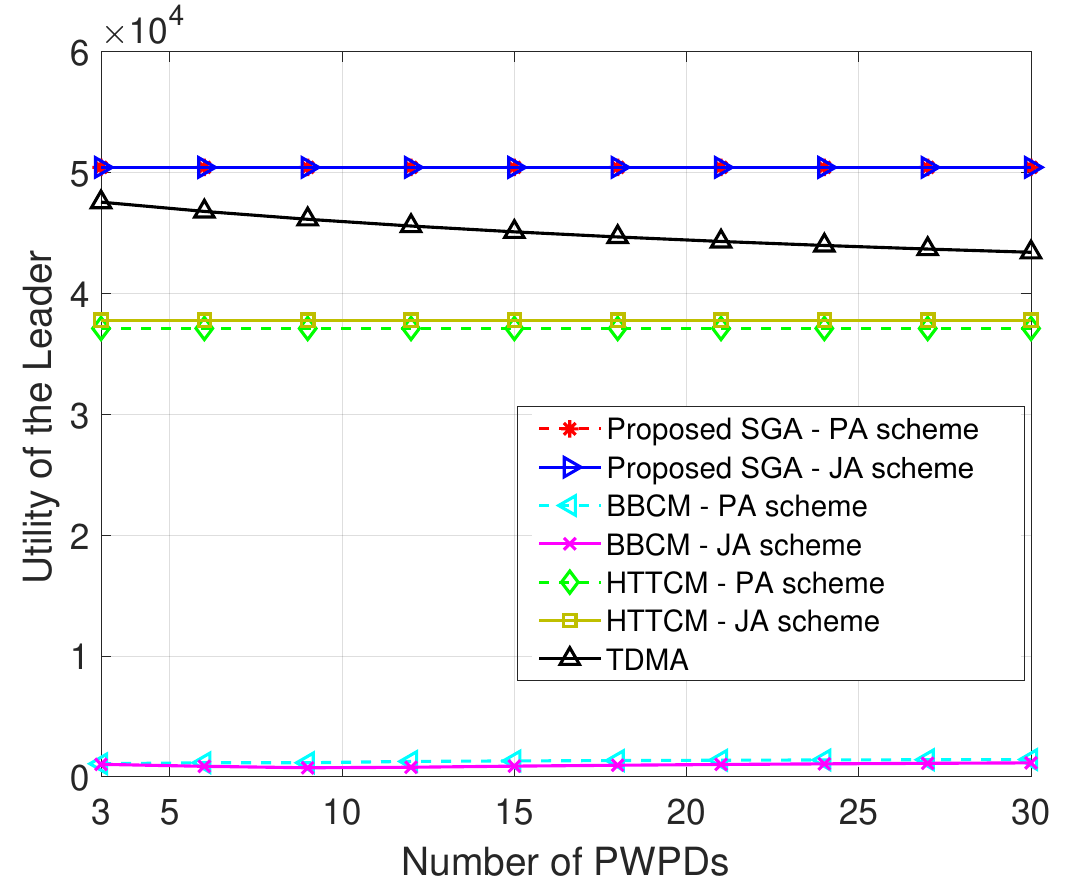} \\ [-0cm]
		(a) & (b) & (c) 
		\end{array}$
		\caption{Leader's payoff under different numbers of (a) HWPDs, (b) AWPDs, and (c) PWPDs.}
		\label{fig:vary_numb_Of_Devices}
	\end{center}
\end{figure*}
\begin{figure*}[!]
	\begin{center}
		$\begin{array}{ccc} 
		\epsfxsize=2.0 in \epsffile{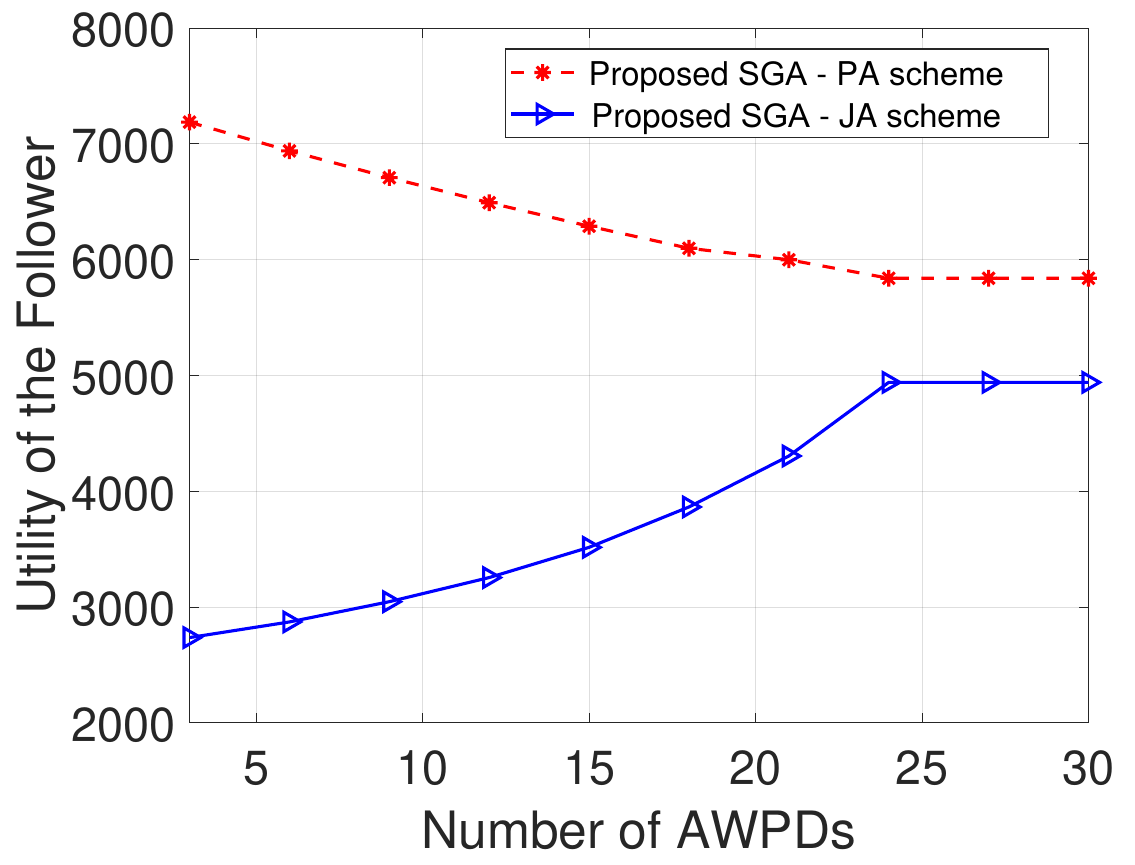} &
		\epsfxsize=2.0 in \epsffile{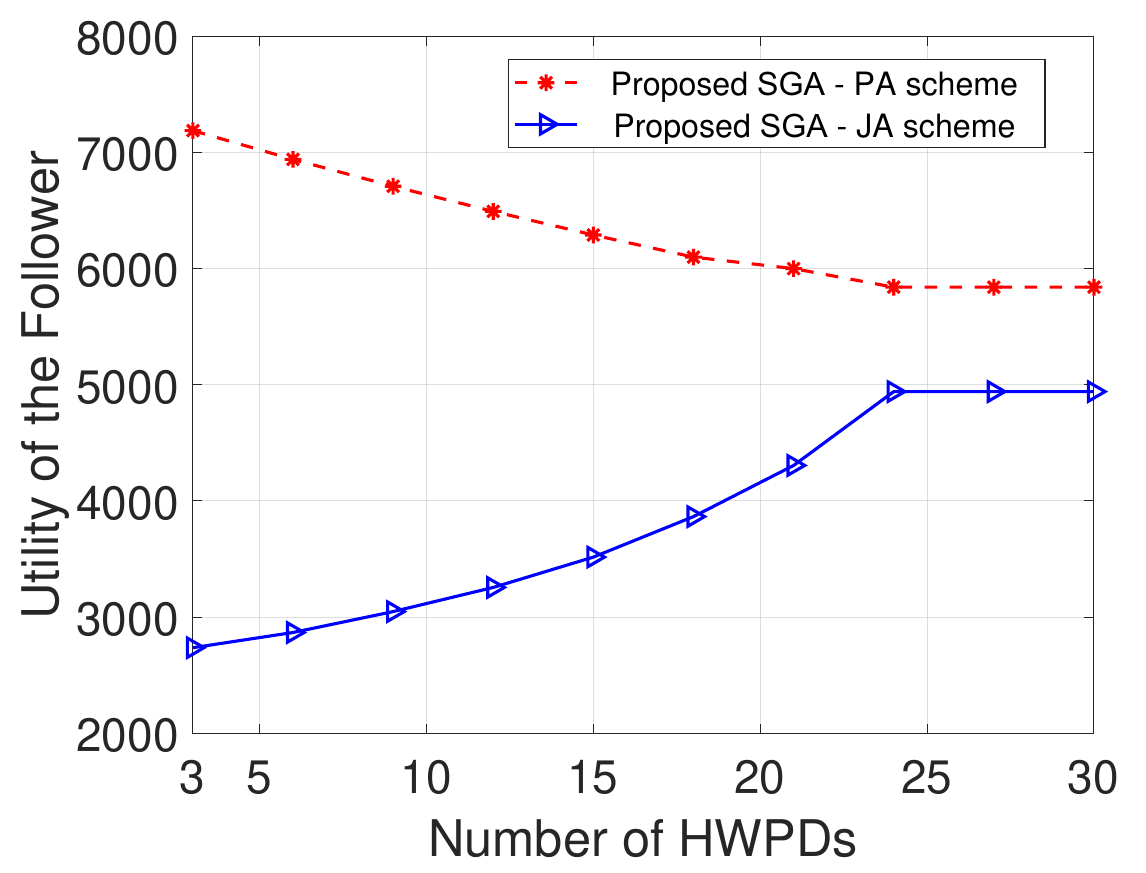} &
		\epsfxsize=2.0 in \epsffile{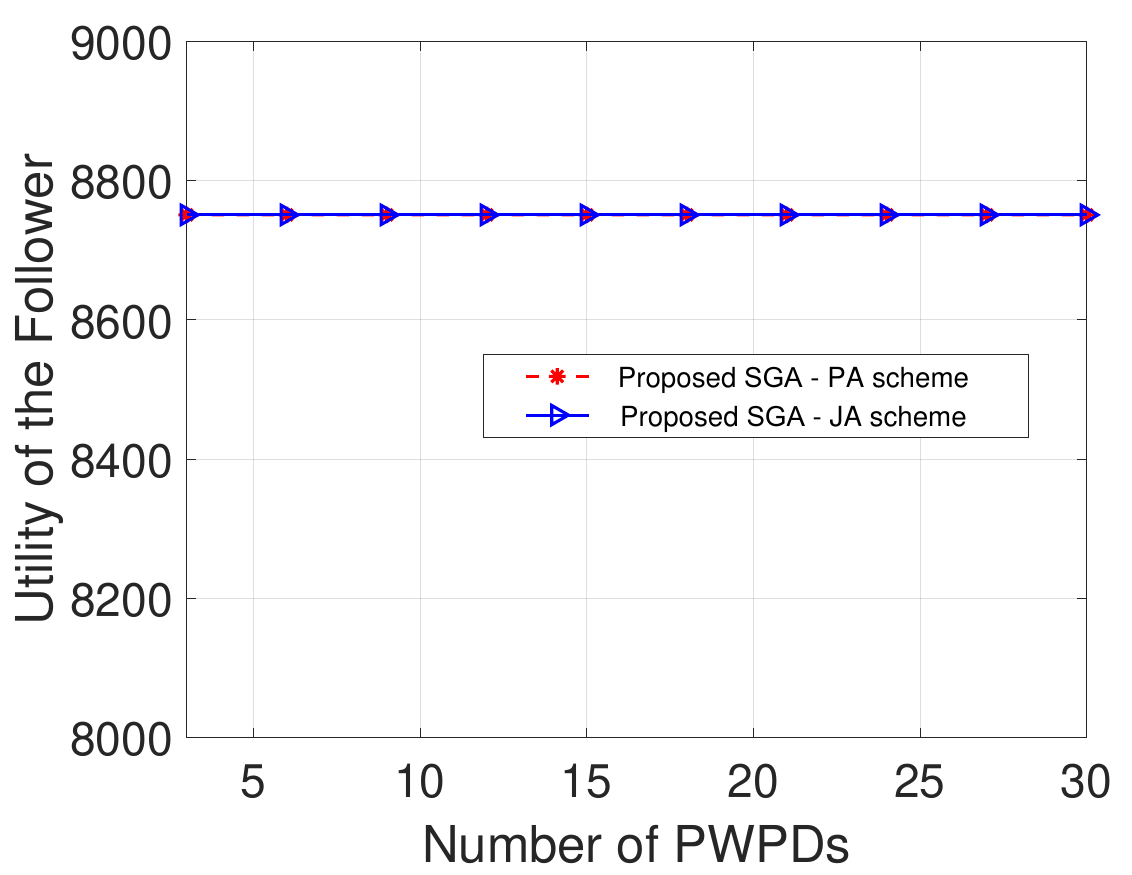} \\ [-0cm]
		(a) & (b) & (c) 
		\end{array}$
		\caption{Follower's payoff under different numbers of (a) HWPDs, (b) AWPDs, and (c) PWPDs.}
		\label{fig:Fol_numbOfIoT}
	\end{center}
\end{figure*}
We now investigate the profits of the ISP and the ESP by altering the number of devices for one type from $3$ to $30$, while fixing the number of devices for other types at $10$. Fig.~\ref{fig:vary_numb_Of_Devices}(a) shows the profit of the ISP achieved by the proposed SGA and the conventional communication methods when varying the number of AWPDs. In general, when the number of AWPDs is small (e.g., $< 24$), the ISP's profit obtained by the proposed SGA using the JA scheme increases and is always highest compared with the others. However, when this number increases (i.e., $\geq 24$), the ISP's profits in the proposed SGA that solved by both the PA/JA schemes are equal. There is no more profit added to the proposed SGA due to the power constraint violation of AWPDs. Moreover, the ISP's profit in the TDMA mechanism is greater than that of the proposed SGA using the PA scheme as the number of AWPDs is small (e.g., $<15$). However, when this number increases, the energy harvesting time in the TDMA mechanism reduces, thus the data throughput obtained by active transmission declines. As a result, the ISP's profit obtained by the TDMA mechanism is smaller than that of the proposed SGA. The ISP's profits obtained by the HTTCM using both the PA/JA schemes shows the same trend but are significantly smaller than those achieved by the proposed SGA. These profits are also smaller than that achieved by the TDMA mechanism. 

The similar trends in the ISP's profits obtained by the proposed SGA, HTTCM, and BBCM, when the number of HWPDs increases, are demonstrated in Fig.~\ref{fig:vary_numb_Of_Devices}(b). The reason is that HWPDs can perform both functions, i.e., backscattering or active transmissions, and the throughput achieved by active transmissions is dominant that from the backscatter communications. By contrast, the ISP's profit earned by the TDMA mechanism is greater than those in other approaches before it remains unchanged when the number of devices is greater than or equal to 9. Finally, as illustrated in Fig.~\ref{fig:vary_numb_Of_Devices}(c), the increase in the number of PWPDs does not impact the ISP's profit obtained by the proposed SGA due to the low backscatter rate. In addition, due to sharing the time resources for PWPDs, the ISP's profit achieved by the TDMA mechanism reduces linearly. 

Fig.~\ref{fig:Fol_numbOfIoT}(a) and Fig.~\ref{fig:Fol_numbOfIoT}(b) show the same trend of the ESP's profit when varying the numbers of AWPDs and HWPDs, respectively. Specifically, the ESP's profit in the case of PA scheme is always higher than that in the case of JA scheme due to the different strategies of the ISP. The strategy of the ISP in the case of PA scheme is to use a high transmission power of the PB in a short time which is opposite to the JA scheme. Moreover, the ESP's profit in the case of PA scheme decreases, while this profit in the case of JA scheme increases before they remain unchanged when the numbers of AWPDs and HWPDs are greater than or equal to 24. The profit of the ESP when varying the number of PWPDs is unchanged as shown in Fig.~\ref{fig:Fol_numbOfIoT}(c) because the ISP does not change its strategy.

\begin{figure*}[!]
	\begin{center}
		$\begin{array}{ccc} 
		\epsfxsize=2.85 in \epsffile{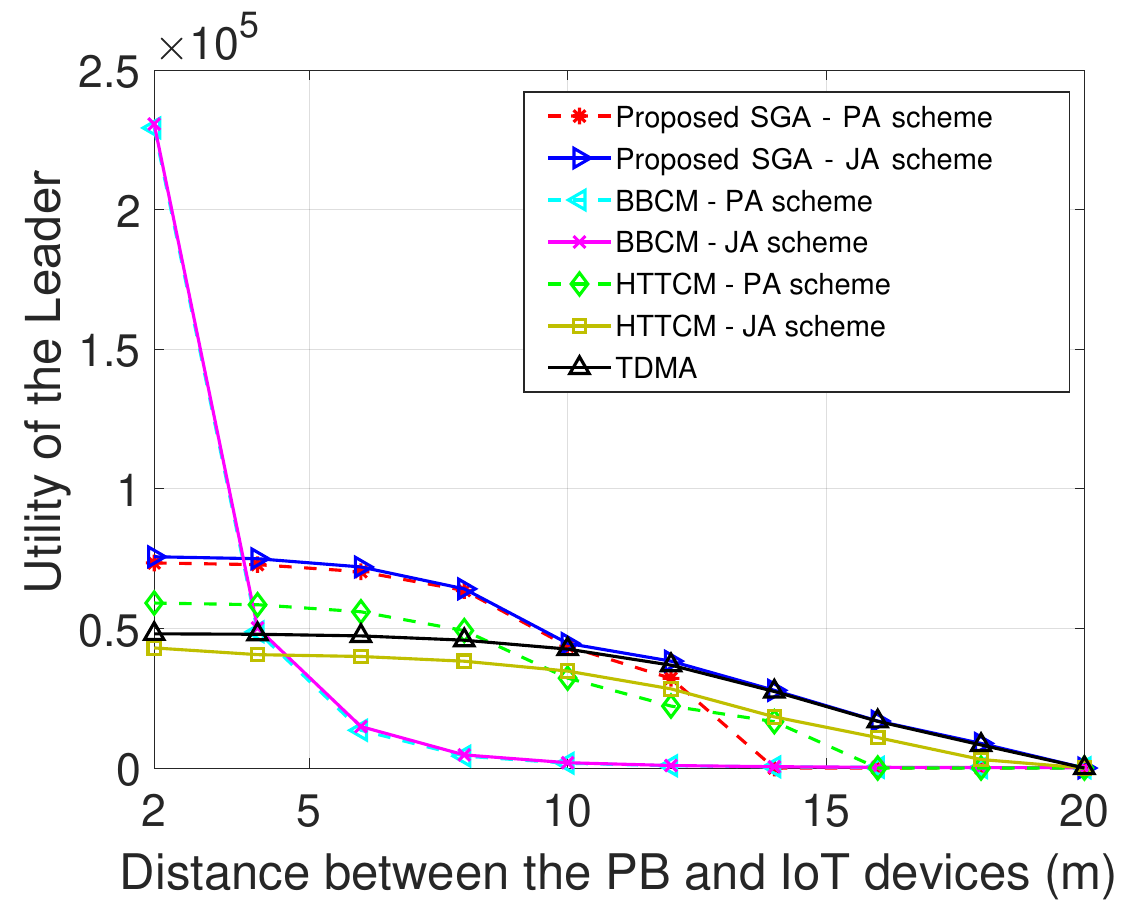}  \quad\quad\quad & 
		\epsfxsize=2.85 in \epsffile{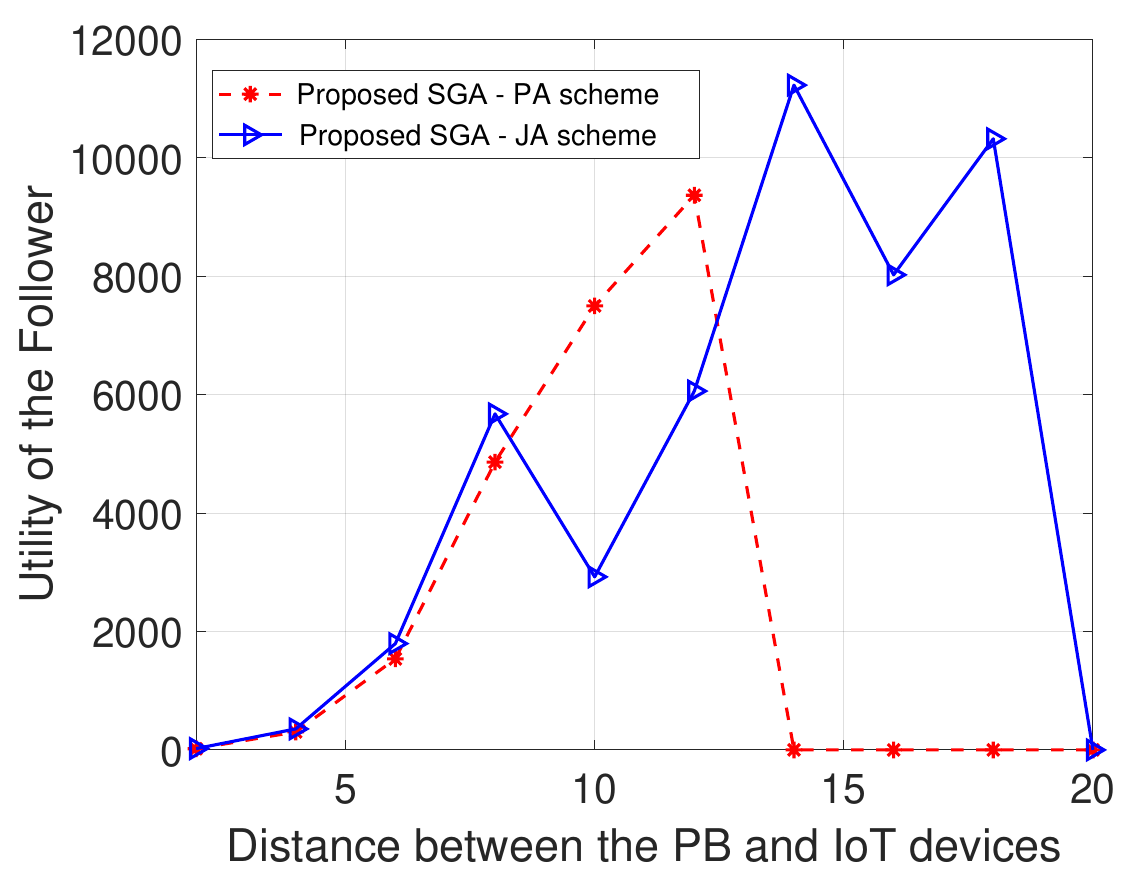} \\ [0.2cm]
		(a) & \quad \quad\quad (b)
		\end{array}$
		\caption{(a) Leader's payoff and (b) Follower's payoff vs. distance between the PB and IoT devices.}
		\label{fig:Varying_dis}
	\end{center}
\end{figure*}

\subsubsection{Impact of Distance between PB and IoT Devices}
Finally, we study the profits of the ISP and ESP achieved by the PA/JA schemes w.r.t. the distance between the PB and the IoT devices in Fig.~\ref{fig:Varying_dis}. 
With the distance of 2 meters, the profits of the ISP obtained by the BBCM using both PA/JA schemes are much greater than those of other approaches as the transmission power of the PB is the major impact on the profit of this approach as demonstrated in Fig.~\ref{fig:Varying_dis}(a). However, its profits drastically decrease as the distance increases. By contrast, the profits of the ISP obtained by the proposed SGA and HTTCM slightly reduce when the distance is smaller than 10 meters. There is no more profit for the proposed SGA and HTTCM using the PA scheme when the distance is greater than or equal to 14 and 16 meters, respectively. Whilst the profits of these approaches solved by the JA scheme are only equal to zero at the distance of 20 meters. It is due to the fact that the achieved profit of the ISP by selling data is smaller than the energy cost when the distance is large. The corresponding profits of the ESP achieved by the PA/JA schemes are plotted in Fig.~\ref{fig:Varying_dis}(b). 
The ESP also obtains no profit when the distance between the PB and IoT devices is higher than 14 meters and 20 meters in the case of PA scheme and JA scheme, respectively, since the ISP quits the game as shown in Fig.~\ref{fig:Varying_dis}(a).

\subsection{Computational Efficiency}
Fig.~\ref{fig: Runtime1000} shows the complexity comparison between the PA and JA schemes in the SGA. Because both schemes take only few iterations to converge, thus in order to compare the computational efficiencies of the proposed schemes more precisely, we measure the runtime of these schemes in 1000 tests with different number of IoT devices for each type (i.e., $N = 5, N = 10, N = 15$). In general, we observe that the runtime of both schemes increases in proportion to the number of IoT devices for each type. Specifically, the maximal runtime needed to solve the proposed SGA by the PA and JA schemes with 15 devices are only 11 seconds and just over 13 seconds on average, respectively. In addition, the computational efficiency of the PA scheme is always better than that of the JA scheme. This is due to the fact that the JA scheme has to run two iterative algorithms (i.e., both inner and outer iterative loops) compared with only one iterative loop of the PA scheme. 

\begin{figure}[t]
	\centering
	\includegraphics[scale=0.38]{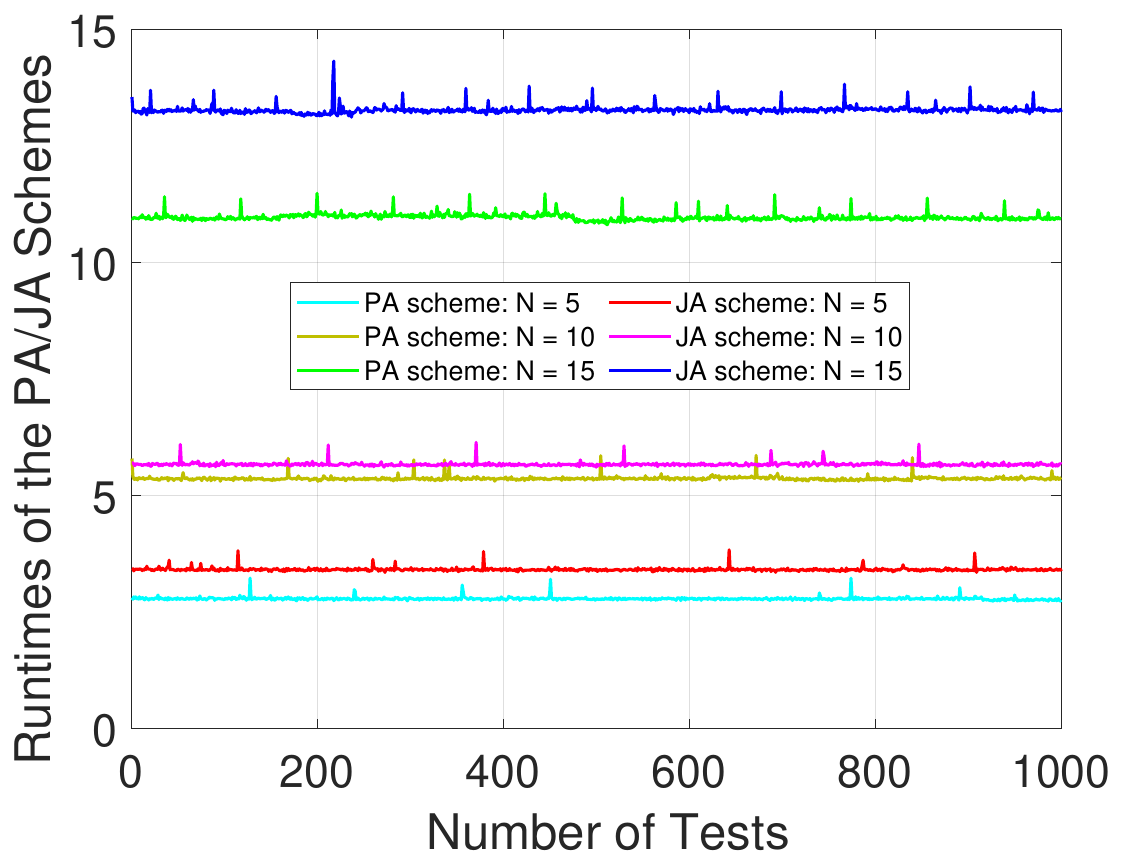}
	\caption{Runtime of proposed schemes.}
	\label{fig: Runtime1000}
\end{figure}

\begin{figure}[t]
	\centering
	\includegraphics[scale=0.38]{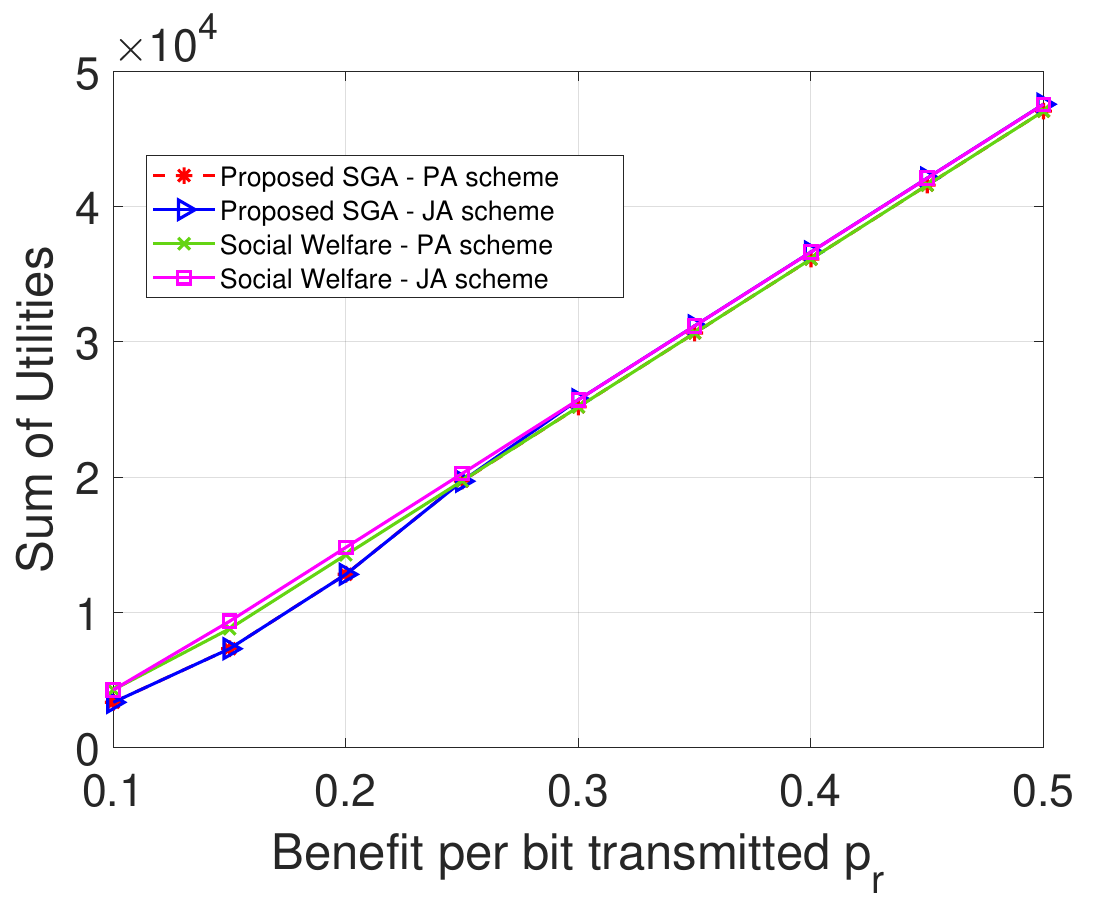}
	\caption{Total profit of providers vs benefit per bit transmitted.}
	\label{fig:Baseline}
\end{figure}
\begin{figure}[t]
	\centering
	\includegraphics[scale=0.41]{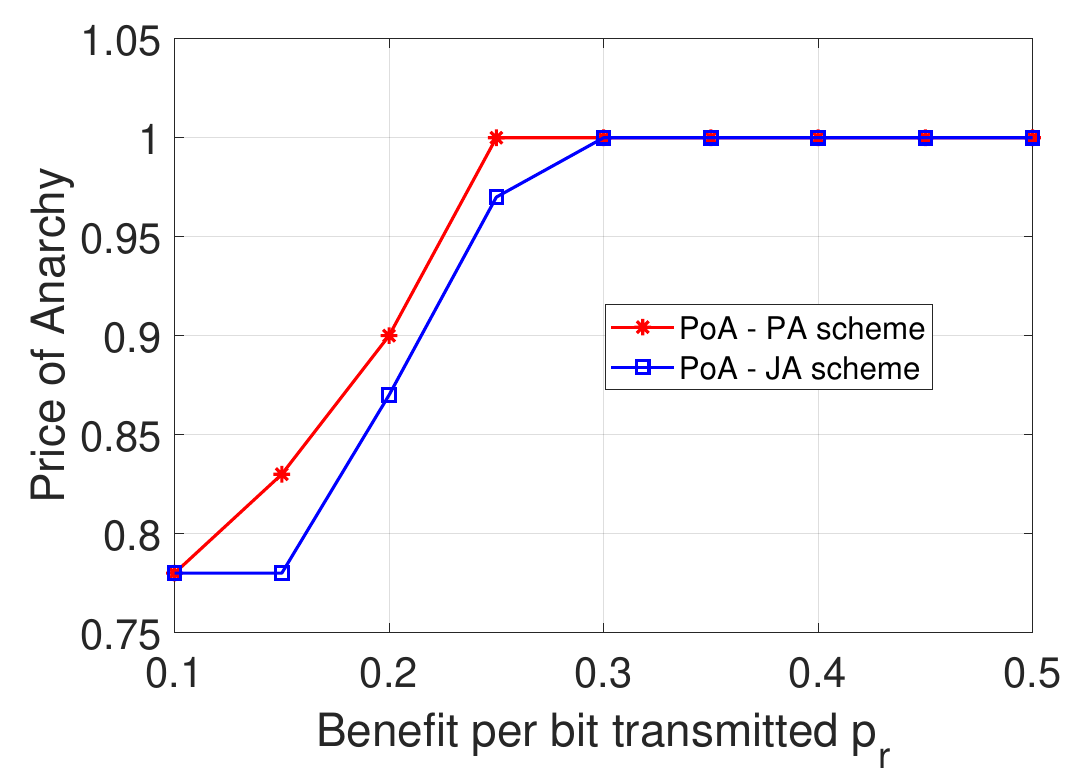}
	\caption{Price of Anarchy.}
	\label{fig:PoA}
\end{figure}

\subsection{Efficiency of the Local Stackelberg Equilibrium}
Fig.~\ref{fig:Baseline} shows the total profits achieved by the proposed SGA and the socially optimal welfare scenario when varying the benefit per bit transmitted $p_r$. It can be observed that the proposed SGA is asymptotic to the social welfare scenario and the gap narrows gradually when the benefit per bit transmitted increases. It is due to the fact that in the proposed SGA, the ISP prefers to purchase less energy from the ESP when the benefit per bit transmitted is low (i.e., lower than 0.25 for the PA scheme and 0.3 for the JA scheme, respectively) to maximize its profit selfishly. Whilst in the socially optimal welfare scenario, the ISP and ESP surrender their selfish behaviors but collaboratively aim to maximize the total benefit (i.e., social welfare), thus the ISP uses more energy to achieve the higher total profit than that of the proposed SGA. In general, both the total profits of the proposed SGA and the SW scenario achieved by the JA scheme are slightly higher than those obtained by the PA scheme.

Fig.~\ref{fig:PoA} shows the PoA ratios which are used to measure the efficiencies of the local SEs obtained by the PA and JA schemes. In both cases, the PoA ratios are always greater than or equal to 0.78, which means a loss of efficiency of 22\% with respect to the social optimum when the benefit per transmitted bit is less than 0.3. When the benefit per bit increases (higher than 0.3), the PoA approaches 1, suggesting that the proposed SGA is socially optimal.

%
\section{Conclusion}
	\label{section:Conclusions}
In this paper, we have jointly optimized the time scheduling and and energy trading to maximize the profits of both the ISP and ESP in heterogeneous IoT wireless-powered communication networks. We have proved the existence and found the local \textit{Stackelberg equilibrium} (SE) that captures the optimal offered price, energy service time, and allocated times for the IoT devices. Simulation results have shown that the proposed Stackelberg game approach solved by the proposed PA/JA schemes always outperform other conventional methods in terms of the ISP's profit. It has also revealed that the JA scheme is superior to the PA scheme in all ISP's performance evaluations. However, the PA scheme yields more profit for the ESP and has better computational efficiency than the JA scheme. Simulations also showed that the obtained local SEs approach the socially optimal welfare when the benefit per transmitted bit is higher than a given threshold.

\appendices

\section{The proof of Theorem~\ref{theorem: SE_existance}}
\label{App:Theorem1}

First, the utility function of the follower in~\eqref{eq: follower_payoff_func} is a quadratic function w.r.t. $P_S$, thus we can obtain a unique optimal solution as shown in~\eqref{eq:followerSol}. Given the strategy of the ISP, the following inequality holds:
\begin{equation}
\label{eq: theorem1_UF}
	{{\bm{U}}_F}\!\left({{P_S^{*}},p_l^*,\beta^*, \bm{{\psi}}^*} \right) \ge {{\bm{U}}_F}\!\left( {{P_S},p_l^*,\beta^*, \bm{{\psi}}^*} \right).
\end{equation}
Second, given the strategy of the ESP, it is straightforward that the constraint set determined by the constraints~\eqref{opt1:a}-\eqref{opt1:f} are compact and the objective function in the problem~\eqref{opt1:main} is a continuous function on this set. According to the well-known Weierstrass Theorem~\cite{Sundaram1996}, the problem~\eqref{opt1:main} admits at least one globally optimal solution. Thus, it implies that the following inequality holds:
\begin{equation}
\label{eq: theorem1_UL}
{{\bm{U}}_L}\!\left({{P_S^{*}},p_l^*, \beta^*, \bm{{\psi}}^*} \right) \ge {\bm{U}_L}\!\left({{P_S^{*}},{p_l},\beta, \bm{{\psi}}}\right).
\end{equation}
Finally, from the inequalities~\eqref{eq: theorem1_UF} and~\eqref{eq: theorem1_UL}, we can conclude that there always exists at least a SE that satisfies the Definition~\ref{definition: SEpoint}.
\section{The proof of Lemma~\ref{lemma: convex_proof_subopt3}}
\label{App:lemma3}

We consider the function  ${\tilde{G}} {(\bm{\psi})}$ in \eqref{eq: funcG3} contributed by four terms $G_p(\bm{\theta}) = \!\sum_{p = 1}^P g_p(\theta_p)$, $G_a(\bm{\nu}) \!= \!\sum_{a=1}^A g_a(\nu_a)$, $G_h(\bm{\tau},\bm{\mu}) \!= \!\sum_{h = 1}^H g_h(\tau_h,\mu_h)$ and a constant $\tilde{C} \!= -{c_3}{p_l^{(n)}\!{\beta^{(n)}}}$ where
\begin{equation}
\begin{aligned}
\left\{\!	\begin{array}{ll}
\!\!{g_p}({\!\theta _p\!})\!\!\!\!&\!=\! {{c_1}{\theta _p}\log_2 \left(1 + {c_3}{\kappa_p}\right)},\\
\!\!{g_a}({\!\nu _a\!})\!\!\!\!&\!=\! {{c_2}{\nu _a}{{\log }_2}\!\left[\! {1\! +\! \frac{{c_3}{\delta_a}{\beta^{(n)}}}{{{\nu _a}}}} \right]},\\
\!\!{g_h}({\!\tau _h},{\mu _h}\!)\!\!\!\! &\!=\!  {c_1}{\tau_h}\!\log_2 \!\left(1 \!+\! {c_3}{\kappa_h}\right) \!+\! {c_2}{\mu_h}\!\log_2\!\!\left[\! {1 \!+\! \frac{{c_3}{\delta_h}\left(\!\beta^{(\!n\!)} - \tau_h\!\right)}{{{\mu _h}}}} \!\!\right]\!\!.
\end{array}\!\!	\right.
\end{aligned}
\end{equation}
It is worth noting that the first term $G_p(\bm{\theta})$ is a sum of linear functions of $\theta_p, \forall p \in \{1,\dots,P\}$.
The second term $G_a(\bm{\nu})$ and third term $G_h(\bm{\psi})$ are sums of concave functions w.r.t. $\nu_a, \forall a \in \{1,\dots, A\}$ and $\tau_h, \mu_h, \forall h \in \{1, \dots, H\}$, respectively, which are straightforward proved by considering their Hessian matrices. Moreover, $\tilde{C}$ is a constant with the fixed $\{p_l^{(n)}, \beta^{(n)}\}$.
Finally, we can conclude that the function $\tilde{G}$ is a concave function w.r.t. $\bm{\psi} \buildrel \Delta \over =  { ( \bm{\theta},  \bm{\nu},  \bm{\tau}, \bm{\mu} )}$ and the problem can be efficiently solved by the interior-point method~\cite{Boyd2004}.

\section{The proof of Theorem~\ref{theorem: convergence_complexity_BCD}}
\label{App: BCD_proof}
We define a constraint set $R(\bm{\chi}) \buildrel \Delta \over = \{R_1(\bm{\chi}), \dots, R_J(\bm{\chi})\}$ determining the feasible region (denoted by $S_{\bm{ \chi}}$) of the problem~\eqref{opt1:main} where $R_i(\bm{\chi})$ is an \textit{i}-th constraint $\forall i \in \{1,\dots,J\}$ and $J$ is the total number of constraints.
Similar to the convergence proof in~\cite{Liao2020}, at $n$-th iteration, from Lemma~\ref{lemma: convex_proof_subopt1},~\ref{lemma: convex_proof_subopt2},~\ref{lemma: convex_proof_subopt3}, we have:
\begin{equation}
\begin{aligned}
\left\{\!	\begin{array}{ll}
{\bm{U}_L}\!{\left({p_l^{(\!n\!)}}\!, {\beta^{(\!n-1\!)}}\!, {\bm{\psi}^{(\!n-1\!)}}\!\right)} \!\ge\! {\bm{U}_L}\!{\left({p_l^{(\!n-1\!)}}\!, {\beta^{(\!n-1\!)}}\!, {\bm{\psi}^{(\!n-1\!)}}\!\right)},\\
{\bm{U}_L}{\left({p_l^{(n)}}\!, {\beta^{(n)}}\!, {\bm{\psi}^{(n-1)}}\!\right)} \!\ge\! {\bm{U}_L}{\left({p_l^{(n)}}\!, {\beta^{(n-1)}}\!, {\bm{\psi}^{(n-1)}}\!\right)}\!,\\
{\bm{U}_L}{\left({p_l^{(n)}}, {\beta^{(n)}}, {\bm{\psi}^{(n)}}\right)} \ge {\bm{U}_L}{\left({p_l^{(n)}}, {\beta^{(n)}}, {\bm{\psi}^{(n-1)}}\right)}.
\end{array}\!\!	\right.
\end{aligned}
\end{equation}

Due to the transitive property, we can obtain ${\bm{U}_L}\!{\left(\!{p_l^{(\!n\!)}}\!, {\beta^{(\!n\!)}}\!, {\bm{\psi}^{(\!n\!)}}\!\!\right)} \!\!\ge\!\! {\bm{U}_L}\!{\left(\!{p_l^{(\!n-1\!)}}\!, {\beta^{(\!n-1\!)}}\!, {\bm{\psi}^{(\!n-1\!)}}\!\!\right)}\!$ which implies that the objective function of the problem~\eqref{opt1:main} is non-decreasing after each iteration.  It is worth noting that the objective function $\bm{U}_L(\bm{\chi})$ is continuous over $R(\bm{\chi})$ which is a compact set. Thus, it is upper-bounded by some finite positive number which means there exists an output $\bm{\hat \chi}^* \buildrel \Delta \over= \bm{\chi}^{(n)}$ satisfying the stopping criterion of \textbf{Algorithm~\ref{algorithm1}}. Therefore, \textbf{Algorithm~\ref{algorithm1}} is guaranteed to be converged.

We then prove $\bm{\hat \chi}^*$ is a locally optimal solution of the problem~\eqref{opt1:main} by following the convergence proof of the BCD method in~\cite{Bertsekas1999}.
It is straightforward that $\bm{U}_{\!L}(\bm{\hat\chi}^*\!) \ge \bm{U}_{\!L}(p_l\!, \beta^{(\!n-1\!)}\!, \bm{\psi}^{(\!n-1\!)}\!), \forall p_l \!\in\! S_{p_l}$ where $S_{p_l}$ is the feasible set of the concave problem~\eqref{subopt1:main}. When ${\lim _{n \to \infty }}\left\| {{\bm{\chi} ^{\left( n \right)}} - {\bm{\chi} ^{\left( {n - 1} \right)}}} \right\| = 0$, we have $\bm{U}_L(\bm{\hat\chi}^*) \ge \bm{U}_L(p_l, \hat {\beta}^*, \bm{\hat \psi}^*), \forall p_l \in S_{p_l}$. Using the K.K.T. conditions for this problem, we obtain that $\nabla_{p_l} L\left( {\hat{p_l}^*,\bm{e}} \right) = 0$ over $S_{p_l}$ where:
\begin{equation}
\nabla_{\!p_l} L\left( {p_l,\bm{e}} \right) = \nabla_{\!p_l} {\bm{U}_{\!L}}\left( p_l  \right) + \sum\limits_{i = 1}^J {e_i}{\nabla_{\!p_l}{R_i}\left( p_l  \right)},
\end{equation}
where $\bm{e} \!\buildrel \Delta \over =\! \{\!e_1\!, \dots, e_J\!\}$ is a Lagrange multiplier vector. Similar repetitions are conducted for the other variable blocks, thus we also have $\nabla_{\!\beta} L\left(\! {\hat {\beta}^*,\bm{e}} \!\right) \!=\! 0$ and  $\nabla_{\!\bm{\psi}} L\left(\! {\bm{\hat \psi}^*,\bm{e}}\! \right) \!=\! 0$ over feasible sets $S_{\beta}$ and $S_{\bm{\psi}}$, respectively. Summarizing these results and using the Cartesian product structure of a feasible set $\hat S_{\bm{\chi}}  \subset S_{\bm{\chi}}$ from $S_{p_l}, S_{\beta}, \text{and } S_{\bm{\psi}}$, we can obtain $\nabla_{\!\bm{\chi}} L\left( {\bm{\hat \chi}^*,\bm{e}} \right) = 0$ over $\hat S_{\bm{\chi}}$ which means that $\bm{\hat\chi}^*$ satisfies the K.K.T. conditions. Moreover, the objective function of the problem~\eqref{opt1:main} is a non-decreasing function over the $\hat S_{ {\bm{\chi}}}$ as the aforementioned discussion. Thus, we can obtain that $\bm{\hat\chi}^*$ is a locally optimal solution of the problem~\eqref{opt1:main}. Then, the proof is completed.

\section{The proof of Theorem~\ref{theorem: convergence_optimal_solution_CCCP}}
\label{App: CCCP_proof}
We first prove the convergence of \textbf{Algorithm~\ref{algorithm2}}. For $k > 1$, we have:
\begin{equation}
\begin{aligned}	
\label{eq: converge_CCCP}
	{\hat{Q}}\!\left( \!{{V^{\left( \!k\! \right)}}}\! \right) & \!\buildrel \Delta \over = \! {Q_{ccav}}\!\left(\! {{V^{\left( k \right)}}}\! \right) \!+\! {Q_{cvex}}\!\left(\! {{V^{\left( k \right)}}} \!\right)\\
	&\ge\! {Q_{ccav}}\!\!\left(\!\! {{V^{\left( \!k\! \right)}}}\!\! \right) \!\!+\! {Q_{cvex}}\!\!\left( \!\!{{V^{\left( \!{k - 1}\! \right)}}}\! \right) \\
	& \quad + {\left(\! {{V^{\left(\! k\! \right)}} \!-\! {V^{\left( \!{k - 1} \!\right)}}} \!\right)^T}\!\!\nabla {Q_{cvex}}\!\!\left(\!\! {{V^{\left( \!{k - 1} \!\right)}}} \!\right)\\
	&\ge\! {Q_{ccav}}\!\!\left(\! {{V^{\left(\! {k - 1} \!\right)}}} \!\!\right) \!+\! {\left( \!{{V^{\left(\! {k - 1}\! \right)}} } \!\!\right)^T}\!\nabla {Q_{cvex}}\!\!\left(\! {{V^{\left(\! {k - 1} \!\right)}}} \!\!\right) \\
	& \quad \!+\! {Q_{cvex}}\!\!\left(\! {{V^{\left(\! {k - 1} \!\right)}}}\!\! \right) \!-\! {\left(\! \!{{V^{\left( \!{k - 1}\! \right)}}}\!\!\right)^T}\!\nabla {Q_{cvex}}\!\!\left( \!{{V^{\left( \!{k - 1} \!\right)}}}\!\! \right)\\
	&=\! {Q_{ccav}}\!\!\left(\! {{V^{\left(\! {k - 1}\! \right)}}}\!\! \right) \!\!+\! {Q_{cvex}}\!\!\left(\! {{V^{\left(\! {k - 1}\! \right)}}}\! \right)  \!\buildrel \Delta \over =\!  {\hat{Q}}\!\left( \!{{V^{\left( \!{k - 1} \!\right)}}} \!\!\right)\!\!,
	\end{aligned}
\end{equation}
where the first inequality in~\eqref{eq: converge_CCCP} is derived from the first order Taylor approximation of a convex function~\cite{Boyd2004}:
\begin{equation}
	{Q_{cvex}}\!\!\left( \!\!{{V^{\left( \!k\! \right)}}} \!\right) \!\ge\! {Q_{cvex}}\!\left( \!{{V^{\left( \!{k - 1} \!\right)}}} \!\right) \!+\! {\left( \!{{V^{\left(\! k\! \right)}} \!-\! {V^{\left(\! {k - 1} \!\right)}}} \!\right)^T}\!\nabla {Q_{cvex}}\!\left(\! {{V^{\left(\! {k - 1} \!\right)}}}\!\! \right)\!.
\end{equation}
The second inequality is obtained from~\eqref{eq: subopt6}. We define ${S_V}$ to be the set of $V$ satisfying the constraints~\eqref{subopt5:a}-\eqref{subopt5:e}. Similar to the proof in the Appendix~\ref{App: BCD_proof} since the function $\hat{Q}(V)$ is continuous on $S_V$ which is a compact set, it is upper-bounded by some positive value when $k$ tends to infinity. Then, the CCCP algorithm will converge to $V^{*}$, i.e., $V^{*} \buildrel \Delta \over = V^{(k)} = V^{(k-1)}$. Thus, \textbf{Algorithm~\ref{algorithm2}} is converged.

Next, we prove that $V^{*}$ is a local optimum of the optimization problem~\eqref{subopt5:main}. We define a constraint set ${R}\left(V\right) \buildrel \Delta \over = \left\{ {{R_1}\left( V \right),{R_2}\left( V \right), \ldots ,{R_J}} \right\}$, where $R_i (V)$ is an \textit{i}-th constraint and $J$ is the total number of constraints in the problem~\eqref{eq: subopt6}. Since $S_V$ is a compact set and ${\tilde{Q}}\left(V\right) = {{Q_{ccav}}\!\left( V \right) + {V^T}\nabla {Q_{cvex}}\!\left( \!{{V^{\left( {k - 1} \right)}}} \right)}$ is a concave function of $V$, we have the K.K.T. conditions for the optimization problem~\eqref{eq: subopt6} as follows:
\begin{equation}
\left\{\!\!\! {\begin{array}{{c}}
	\!{\!\nabla {Q_{ccav}}\!\!\left(\! {{V^{\left( \!k \!\right)}}} \!\right) \!+\! \nabla {Q_{cvex}}\!\left(\! {{V^{\left(\! {k - 1}\! \right)}}}\! \right) \!+\! {Y^T}\nabla R\!\left( \!{{V^{\left( \!k\! \right)}}}\!\right) \!=\! 0},\\
	\!{Y \!\!=\!\! \left[ {{y_1}, \ldots ,{y_J}} \right]\!,{y_i} \!\ge\! 0,{y_i}{R_i}\!\left( \!{{V^{\left(\! k \!\right)}}}\! \right) \!=\! 0,\forall i \!=\! \left\{\! {1, \ldots ,J}\! \right\}}\!,
	\end{array}} \right.
\end{equation}
where $Y$ is the optimal Lagrangian variable set for $V^{(k)}$. When $V^{(k)} = V^{(k-1)} = V^{*}$, the above equation set can be rewritten as follows:
\begin{equation}
\left\{\!\! {\begin{array}{*{20}{c}}
	{\nabla {Q_{ccav}}\left( {{V^*}} \right) + \nabla {Q_{cvex}}\left( {{V^*}} \right) + {Z^T}\nabla R\!\left( {{V^*}} \right) = 0},\\
	{\!\!\!Z \!\!=\!\! \left[ {{z_1}, \ldots ,{z_J}} \! \right]\!,{z_i} \!\ge\! 0,{z_i}{R_i}\!\left(\! {{V^*}} \!\right) \!=\! 0,\forall i\!=\! \left\{ \!{1, \ldots ,J} \!\right\}},
	\end{array}} \right.
\end{equation}
where  $Z$ is the optimal Lagrangian variable set for $V^{*}$. Thus, $V^{*}$ also satisfies the K.K.T. conditions for the problem~\eqref{subopt5:main}. Moreover, from the inequality~\eqref{eq: converge_CCCP}, we have the $Q(V)$ is a non-decreasing function over $S_V$. Therefore, $V^{*}$ is a local optimum of the problem~\eqref{subopt5:main}. The proof is similar to~\cite{Feng2015}.



\end{document}